\documentclass[11pt]{article}
\usepackage[utf8]{inputenc}

\usepackage[margin=1in]{geometry}

\usepackage{float,graphicx,verbatim,fullpage,hyperref,amssymb,amsmath,amsthm,enumerate,multicol, xspace,xcolor,mathtools,thmtools,thm-restate,cleveref,xspace,tikz,caption,subcaption,wasysym, enumitem}
\usepackage[margin=1in]{geometry}
\usepackage{algorithm, cite}
\usepackage[noend]{algpseudocode}
\usepackage{enumitem}
\usepackage{comment}

\usetikzlibrary{calc}
\usetikzlibrary{positioning}

\usepackage{mleftright}
\usepackage{hyperref}
    \hypersetup{colorlinks=true, linkcolor=blue, filecolor=magenta, urlcolor=blue, citecolor=red}

\usetikzlibrary{arrows.meta}
\tikzset{>={Latex[width=1.5mm,length=1.5mm]}}

\newfloat{procedure}{htbp}{loa}
\floatname{procedure}{Procedure}

\def\R{\mathbb{R}}
\def\Q{\mathbb{Q}}

\def\Z{\mathbb{Z}}

\newcommand{\opt}{\textsf{OPT}}

\def\ep{\varepsilon}
\def\tO{\tilde{O}}

\newtheorem{theorem}{Theorem}[section]

\newtheorem{lemma}[theorem]{Lemma}
\newtheorem{claim}[theorem]{Claim}

\theoremstyle{definition}
\newtheorem{definition}[theorem]{Definition}

\providecommand{\email}[1]{\href{mailto:#1}{\nolinkurl{#1}\xspace}}

\def\final{0}  
\def\iflong{\iffalse}
\ifnum\final=0

\else
\fi

\usepackage{environ}

\newcommand{\eps}{\varepsilon}

\def\*#1{\mathbf{#1}}
\def\+#1{\mathcal{#1}}

\NewEnviron{problem}[1]{
	\begin{center}\fbox{\parbox{6in}{
				{\centering\scshape #1\par}
				\parskip=1ex
				\everypar{\hangindent=1em}
				\BODY
}}\end{center}}

\newcommand{\poly}{\ensuremath{\mathsf{poly}}}
\newcommand{\polylog}{\ensuremath{\mathsf{polylog}}}

\makeatletter
\newcommand*{\inlineequation}[2][]{
  \begingroup
    \refstepcounter{equation}
    \ifx\\#1\\
    \else
      \label{#1}
    \fi
    \relpenalty=10000 
    \binoppenalty=10000 
    \ensuremath{
      #2
    }
    ~\@eqnnum
  \endgroup
}
\makeatother
\newcommand{\gsf}{\textsc{Group Steiner Forest}\xspace}
\newcommand{\gsr}{\textsc{Group Steiner Rounding}\xspace}

\newcommand{\rcsp}{\textsc{Resource-constrained Shortest Path}\xspace}

\newcommand{\gbh}{\textsc{Generalized $\beta$-Hopset}\xspace}

\newcommand{\mcs}{\textsc{Multicriteria Spanner}\xspace}
\newcommand{\pks}{\textsc{Packing-Covering Spanner}\xspace}
\newcommand{\rcs}{\textsc{Routing-controlled Spanner}\xspace}

\newcommand{\gss}{\textsc{Graph spanner for Group Steiner metric}\xspace}

\newcommand{\Dist}{\textsc{Dist}}

\newcommand{\spks}{\textsc{Scaled Packing-Covering Spanner}\xspace}

\newcommand{\hop}{\textsc{Hop-bound}\xspace}

\newcommand{\mslc}{\textsc{Minimum Density Steiner Label Cover}\xspace}
\newcommand{\bdgt}{\textbf{Bdgt}\xspace}
\newcommand{\ctrl}{\textbf{Ctrl}\xspace}

\newcommand{\dis}{\textsc{dis}\xspace}

\newcommand{\threshold}{\tau}

\newcommand{\res}{\boldsymbol{RES}\xspace}
\newcommand{\sres}{\textsc{ScaledRes}\xspace}

\newcommand{\zvec}{\boldsymbol{0}\xspace}

\DeclareMathOperator{\siign}{sgn}

\newcommand{\lenup}{\textsc{Max length}\xspace}
\newcommand{\lenlow}{ \textsc{Min length}\xspace}

\newcommand{\rcjt}{\textsc{Minimum-density Resource-constrained Junction Tree}\xspace}
\newcommand{\trcjt}{$\theta$-\textsc{Relaxed Minimum-density Resource-constrained Junction Tree}\xspace}

\newcommand{\cost}{\textsc{Cost}\xspace}

\newcommand{\rrcjt}{\textsc{$\theta$-relaxed Minimum-density Resource-constrained Junction Tree}\xspace}
\newcommand{\validlayer}{\textsc{VL}\xspace}

\newcommand{\indexlist}{\textsc{A}\xspace}

\newcommand{\cJ}{\mathcal{J}}
\def\ep{\varepsilon}
\def\tO{\tilde{O}}

\title{Routing-Controlled Spanners}

\author{
Elena Grigorescu\thanks{University of Waterloo.
 E-mail: \email{elena.grigorescu@uwaterloo.ca}. Supported in part by NSF CCF-1910659, NSF CCF-1910411, and NSF CCF-2228814.
}
 \and
 Nithish Kumar\thanks{Purdue University. E-mail: \email{kumar410@purdue.edu}. Partially supported by NSF CCF-1910411, NSF CCF-2228814, NSF Award CCF-2127806 and ONR Award N00014-24-1-2695.}
 \and
 Young-San Lin\thanks{Melbourne Business School. 
 E-mail: \email{y.lin@mbs.edu}.}
}   

\begin{document}

\maketitle

\begin{abstract}
\normalsize
Designing sparse directed {\em spanners}, which are subgraphs that approximately maintain distance constraints, has attracted sustained interest in TCS, especially due to their wide applicability, as well as the difficulty to obtain tight results. However, a significant drawback of the notion of spanners is that it does not capture a natural setting where demand pairs are subject to restrictions beyond a distance constraint.

In this paper we initiate the study of \emph{routing-controlled spanners},  where in addition to distance constraints, demand pairs are also subjected to routing constraints, which may require or forbid visiting specific vertices on feasible paths. The goal is to find a minimum-cost routing solution that satisfies the multiple constraints. Moreover, we introduce an even more general notion, which we call \emph{packing-covering spanner}\footnote{In a previous version of this article, we referred to \pks as \mcs. However, we have changed the name as we believe that the new name is more appropriate for our model.}, where  each demand pair is subjected to a number of packing and covering constraints, in addition to a distance constraint. Packing-covering spanners capture other natural network connectivity problems such as optimal hopsets and graph spanners for group Steiner distances.

To the best of our knowledge, we obtain the first approximation algorithms for the packing-covering spanner problem, and thus for the routing-controlled spanners, under natural assumptions. Our results match the state-of-the-art approximation ratios in special cases of ours, such as Steiner Forests and Directed Spanners. Our results also imply approximation algorithms for optimal hopsets and graph spanners for group Steiner distances in the directed setting, and position packing-covering spanners as a natural abstraction unifying several well-studied problems in directed network design.

Our main technical tool is a delicate generalization of the minimum-density junction tree framework of Chekuri, Even, Gupta, and Segev (TALG 2011) to the notion of \emph{minimum-density resource-constrained junction trees}, which also extends ideas from  Chlamt{\'a}{\v{c}},  Dinitz, Kortsarz, and Laekhanukit (TALG 2020).

\end{abstract}

\section{Introduction}
\label{sec:intro}

A fundamental goal in network design problems is to cheaply connect terminal pairs, while satisfying distance or connectivity constraints between demand pairs. In particular, the notion of graph {\em spanners} \cite{Awerbuch, PelegS89} has been central to the study of multi-commodity network problems. Spanners are subgraphs of a given graph network, in which pairwise distances approximate the respective distances in the original graph. For the past several decades, they have found a wide range of applications both in theory and practice, including  distributed computation \cite{Awerbuch, PelegS89}, data structures \cite{Yao1982SpacetimeTF, Alon87optimalpreprocessing}, routing schemes \cite{PelegU89a,CowenW04,RodittyTZ08,PachockiRSTW18}, approximate shortest paths \cite{DorHZ00, Elkin05, BaswanaK10}, distance oracles \cite{BaswanaK10, Chechik15, PatrascuR14}, and  property testing and reconstruction \cite{bhattacharyya2012transitive, BhattacharyyaGJJRW12, BermanBGRWY14, AwasthiJMR16}. See also more comprehensive surveys \cite{Raskhodnikova10, ahmed2020graph}.

Despite being so   widely applicable, a significant drawback of the notion of spanner is that it does not capture a natural setting where demand pairs are subject to restrictions beyond a distance constraint. For example, in transportation networks, not every route is accessible to every vehicle \footnote{Consider for instance, a huge pickup truck that can't fit inside tunnels or lightweight bridges}. A similar, yet distinct scenario is when the route of some vehicles explicitly needs to go through specific points \footnote{Consider for instance, the same pickup truck which needs to pick up its goods at specific points}. Situations such as these are very common in everyday life; however, current definitions of spanners and other network design problems cannot accommodate such restrictions. To address this drawback, we introduce and study the notion of \emph{routing-controlled} spanners in directed graphs, as a general framework that captures the usage of restrictions on routes connecting demand pairs. In addition to distance constraints, demand pairs are also associated with vertex constraints which may disallow or force a path resolving the demand pair to go through certain vertices. 

We now formalize the notion of \rcs and then we discuss its connection to spanners. 

\begin{restatable}{definition}{defrcs} 
\label{def:rcs}
    
    \rcs
    
    \textbf{Instance}: A directed graph $G = (V, E)$, and 
    \begin{enumerate}
        \item \textbf{General parameters:} Number of controls $m\in \Z_{> 0}$ ($m$ is a constant) and number of nodes $n$. 
    
        \item  \textbf{Edge cost:} Each edge $e \in E$ has a cost $\sigma(e)$, given by $\sigma: E \to \Q_{\ge 0}$. We define the cost of a subgraph $H$ of $G$ as $\cost(H) = \sum_{e \in H} \sigma(e)$. 
        
        \item \textbf{Edge length:} Each edge $e \in E$ also has a length $\ell_e \in [\poly(n)]_+$.

        \item \textbf{Groups:} We are also given $m$ sets of vertices $S_1,S_2,\ldots,S_c, \ldots, S_{c+p} \subseteq V$ (where $c+p = m$). We assume that the last $p$ groups have constant size.
        
        \item \textbf{Demand Pairs:} We are given a set $D \subseteq {V} \times {V}$ of ordered pairs. The set $D$ is also called as the demand pair set.
        
        \item \textbf{Controls:} We define the function $\ctrl : D \to \Z^+ \times \{0,1\}^c \times \{-1,0\}^p$ in the following way. $\ctrl(s,t)$ is a vector of $m+1$ elements called the {\em control vector}.
        \begin{itemize}
            \item For $(s,t) \in D$,  The zeroth element, $\ctrl(s,t)[0]$ captures the distance constraint. 
            \item The next $c$ elements  $\ctrl(s,t)[i]$ where $i \in [1,c]$ are either $1$ or $0$; these elements are set to $1$ if a walk resolving $(s,t)$ must visit at least one vertex in $S_i$ and $0$ if there are no such requirements.
            \item The last $p$ elements  $\ctrl(s,t)[i]$ where $i \in [c+1,p+c]$ are either $-1$ or $0$; these elements are set to $-1$ if a walk resolving $(s,t)$ must avoid all vertices in $S_i$ and $0$ if there are no requirements \footnote{Given a specific vector $\ctrl(s,t)$, we use the $0^{th}$ element of the vector i.e., $\ctrl(s,t)[0]$ to capture the distance constraint; the last $m$ elements are used for the routing controls.}. 
        \end{itemize}
         Furthermore, we assume that there exists a {\em routing-controlled feasible} $s \leadsto t$ walk (see Definition \ref{def:group_feasible_walk}) for each $(s,t) \in D$ that satisfies the control vector for each $(s,t) \in D $ \footnote{One can in fact detect whether or not such a walk exists by using the extension of \rcsp that is presented in \cite{grigorescu2024directed}. The problem is not defined when no such walk exists.}.          
        
    \end{enumerate}   

    \textbf{Objective}: Find a minimum cost  subgraph $H$ that satisfies all demands $(s,t) \in D$. That is, for all $(s,t) \in D$ there exists a {\em routing-controlled feasible} $s \leadsto t$ walk $p(s,t)$.
\end{restatable}

\begin{restatable}[Routing-controlled feasible walk]{definition}{defroutingfeasiblewalk}\label{def:group_feasible_walk}
For a given demand pair $(s,t) \in D$, let $\dis(s,t) = \ctrl(s,t)[0]$ \footnote{For a vector $\boldsymbol{v}$, we use $\boldsymbol{v}[i]$ to denote the $i$-th dimensional entry value.}. We say that a $s \leadsto t$ walk $p(s,t)$ is {\em Routing-control-feasible} if 
\begin{itemize}
    \item $\sum_{e \in p(s,t)} \ell_e \leq \dis(s,t) = \ctrl(s,t)[0]$ and
    \item For all $i \in [1,c]$, if $\ctrl(s,t)[i] = 1$, $p(s,t)$ touches at least one point in $S_{i}$.
    \item For all $i \in [c+1,m]$, if $\ctrl(s,t)[i] = -1$, $p(s,t)$ does not touch any points in $S_{i}$.
\end{itemize}
\end{restatable}

\paragraph{Discussion about Definition \ref{def:rcs}:}

Note that in this problem, we could have two different must-visit vertices for the same $(s,t)$ and the only way to visit them both may force us to go through the same intermediate point twice. For instance, consider the graph $G$ in Figure \ref{fig:walk_graph}. If we have a demand pair $(a,e)$ whose route must through $h$ and $g$, it has to visit $c$ at least twice. For this reason, we need to consider walks as opposed to paths for \rcs.

\begin{figure}[ht]
\centering
\begin{tikzpicture}[>=Stealth,  
    every node/.style={circle, draw, inner sep=1.5pt, minimum size=7mm, font=\small}
]
\node (v1) at (0,0) {a};
\node (v2) at (1.6,0) {b};
\node (v3) at (3.2,0) {c}; 
\node (v4) at (4.8,0) {d};
\node (v5) at (6.4,0) {e};
\draw[->] (v1) -- (v2);
\draw[->] (v2) -- (v3);
\draw[->] (v3) -- (v4);
\draw[->] (v4) -- (v5);

\node (u1) at (3.2,1.6)  {f};
\node (u2) at (4.4,0.8)  {g};
\draw[->] (v3) -- (u1);
\draw[->] (u1) -- (u2);
\draw[->] (u2) -- (v3);

\node (l1) at (3.2,-1.6) {h};
\node (l2) at (2.0,-0.8) {i};
\draw[->] (v3) -- (l1);
\draw[->] (l1) -- (l2);
\draw[->] (l2) -- (v3);

\end{tikzpicture} 
\caption{A graph $G$ where any feasible solution for the given demand pair $(a,e)$ requires revisiting vertex $c$ twice if both $h$ and $g$ are must-visit vertices. This highlights the need to allow walks (rather than only simple paths) when defining routing-controlled feasibility.}
\label{fig:walk_graph}
\end{figure}

Using Definition \ref{def:rcs}, we can prescribe, for a walk $p$, the set of vertices it must pass through as well as the set of vertices it must avoid. One could use $\ctrl(s,t)[i]$ where $i \in [1,c], i \in Z^+$ to force any walk resolving $(s,t)$ to go through $c$ different sets of points \footnote{Returning to our earlier example about pickup trucks, this allows us to force the truck to pickup raw materials from specific points; such as fuel from a gas station, paint from a hardware store.}. One can then use $\ctrl(s,t)[i]$ where $i \in [c+1,p+c]$ to force any walk resolving $(s,t)$ to avoid some points \footnote{So, if $(s_1,t_t)$ provides a walk for a truck, it can avoid tunnels using the last $p$ elements of $\ctrl(s_1,t_1)$; while a different $(s_2,t_2)$ which provides a walk for cars, need not avoid them.}. The special constraint $\ctrl(s,t)[0]$ gives the distance constraint for the demand pair $(s,t) \in D$. 

As mentioned above, \rcs captures a very natural scenario that has not been addressed by the existing literature. In this work, we provide the first approximation algorithm for \rcs in the directed setting. A weaker version of \rcs where we only need to handle a walk with must-visit vertices (and not worry about walks avoiding some vertices) was explored in \cite{bilo2024graph} for the undirected setting. That version is called \gss and our work here gives an approximation algorithm for \gss in the directed setting whose approximation factor matches that of some subproblems like Steiner forest. 

\paragraph{\pks:}
Although \rcs already captures a natural setting not studied before, it is a special case of an even  more general problem that we introduce and study, called the \pks problem. In the \pks problem, every demand pair can be associated with a number of packing and covering constraints in addition to a distance constraint. We give the first approximation algorithms for \pks. We will see later on how \pks can be useful as a general tool by using our results for \pks to recover and slightly generalize a result for designing hopsets from \cite{dinitz2025approximation} (see Section \ref{sec:gbh}).

We now formalize the notion of \pks. Some readers may notice the strong similarities between \pks and \rcs. We have separated these two problems for ease of presentation. We believe that \rcs is a very natural formulation of \pks, which allows us to understand the impact and utility of the problem in the real world. \pks, on the other hand, is a distilled version of \rcs that is easier to deal with mathematically. 
 
\begin{restatable}{definition}{defpks} 
\label{def:pks}
    
    \pks
    
    \textbf{Instance}: A directed graph $G = (V, E)$, and 
    \begin{enumerate}
        \item \textbf{General parameters:} Magnitude threshold $\threshold \in \Z_{\ge 0}$, and number of resources $m+1 \in \Z_{> 0}$. 
    
        \item  \textbf{Edge cost:} Each edge $e \in E$ has a cost $\sigma(e)$, given by  $\sigma: E \to \Q_{\ge 0}$. We define the cost of a subgraph $H$ of $G$ as $\cost(H) = \sum_{e \in H} \sigma(e)$. 
        
        \item \textbf{Edge resource consumption:} Each edge $e \in E$ has an $(m+1)$-dimensional resource consumption vector $\boldsymbol{r_e} \in \R \times [\threshold]_{+}^p \times [\threshold]_{-}^c$ \footnote{For any $\tau \in \Z_{\ge 0}$, we define $[\tau]_{+}:=\{0,1,...,\tau\}$ and $[\tau]_{-}:=\{-\tau,-\tau+1,...,-1,0\}$.} where $c+p = m$ and $c,p \in Z^+$.  
        The resource $\boldsymbol{r_e}[0]$ is called the length of an edge $e$; we also use $\ell_e$ to denote $\boldsymbol{r_e}[0]$. We assume that there are no negative cycles for length.
        
        \item \textbf{Demand Pairs and Budgets:} A set $D \subseteq {V} \times {V}$ of ordered {\em demand} pairs, and a resource {\em budget} captured by the function $\bdgt: D \to \R \times [\threshold]_{+}^p \times [\threshold]_{-}^c$. Furthermore, we assume that there exists a feasible (see Definition \ref{def:feasible_walk}) $s \leadsto t$ walk\footnote{We do allow negative cycles for the last $c$ resources - due to this reason, we may need to consider walks as opposed to simple walks. However, when the budget of the last $m$ resources is a constant, these walks will only have polynomially many edges.} that satisfies the resource budget for each $(s,t) \in D$ \footnote{One can in fact detect whether or not such a walk exists by using the extension of \rcsp that is presented in \cite{grigorescu2024directed}. The problem is not defined when no such walk exists.}. 
        
    \end{enumerate}   

    \textbf{Objective}: Find a minimum cost  subgraph $H$ that satisfies all demands $(s,t) \in D$. That is, for all $(s,t) \in D$ there exists a {\em feasible} $s \leadsto t$ walk $p(s,t)$.
\end{restatable}

\begin{restatable}[Feasible walk]{definition}{deffeasiblewalk}\label{def:feasible_walk}
 For a given demand pair $(s,t)$, we say that a $s \leadsto t$ walk $p(s,t)$ is {\em feasible} if $\sum_{ e \in p(s,t)} \boldsymbol{r_{e}}[i] \le \bdgt(s,t) [i] \: \forall i \in [m+1]$.
\end{restatable}

\paragraph{Discussion about Definition \ref{def:pks}:}

We remark that for a resource where $\boldsymbol{r_{e}}[i] \geq 0 \:, \: \forall e \in E$, the budget $\bdgt(s,t)[i]$ can be interpreted as a {\em packing constraint}. Thus, we call such a resource as a {\em packing resource}. In contrast, for a resource where $\boldsymbol{r_{e}}[i] \leq 0 \:, \: \forall e \in E$, the budget $\bdgt(s,t)[i]$ can be interpreted as a {\em covering constraint}. Thus, we call such a resource as a {\em covering resource}.

We note that unit-cost spanners \cite{chlamtavc2020approximating,berman2013approximation,grigorescu2021online} are a special case of \pks, where we have only one resource and $\sigma(e) = 1, \forall e \in E$. In addition, the notion of \pks also generalizes weighted spanner \cite{gkl2023},  which is in turn a unification of Steiner forest\footnote{A Steiner forest is a min cost subgraph that preserves connectivity between pairs of interest. It has no distance constraints.} \cite{abboud2018reachability,berman2013approximation,feldman2012improved,chekuri2011set} and unit-length spanners \cite{chlamtavc2020approximating}.

We are also not aware of any previous results for the \pks problem when $m \geq 1$, except for one very special case. When $|D| = 1$ (i.e., when we have only one demand pair $(s,t)$ and just need to find a $(s,t)$ walk), \cite{horvath2018multi} presents an approximation algorithm in which all resource constraints are violated up to a factor of $1+\eps$,  while the cost is optimal. It also assumes that the resource consumptions are positive for all resources, but does not place any restriction on the budget. The results of \cite{grigorescu2024directed} generalize \cite{horvath2018multi} by allowing the resource consumption to be negative, but restrict the range of the negative resource consumptions.

\subsection{Our contributions} \label{sec:intro-con}

Our most general results address the case when the edge lengths (the zeroth resource) may have rational or negative values. However, for clarity, we first describe the case where the length is a positive integer. We focus on \pks first and move on to \rcs afterwards.

\subsubsection{Our results for \pks}

The following theorem assumes that the length is a positive integer. Throughout this work, we fix $k= |D|$.

Recall that in \pks, the last $m$ resources are always bounded integers.
The formal proof of the following theorem appears in Section \ref{sec:khalf_vanilla}.

\begin{restatable}{theorem}{thmvpkskhalf} \label{thm:vpks_khalf} (Integer version)
   If $\boldsymbol{r_e}[0] \in \{1, 2, \ldots, T\}$ for all $e\in E$, then for any two constants $\ep>0, m>0$, there exists a $\tO(k^{1/2 + \ep} \cdot \threshold^m)$-approximation algorithm for \pks running in time $\poly(n,\tau, T)$.
\end{restatable}

As previously mentioned, Theorem \ref{thm:vpks_khalf} addresses a more general problem than previous literature, yet it recovers some results of the special cases of weighted spanners \cite{gkl2023}, unit-cost spanners \cite{grigorescu2021online}, and Steiner forests \cite{chekuri2011set} \footnote{We note that previous results in directed network design typically come in two flavors. The approximation factor for these results is either a function of $n$ or a function of $k$. This work only focuses on results whose approximation factor is a function of $k$.}. 

We also prove a more general fractional version of Theorem \ref{thm:vpks_khalf} that allows us to handle {\em rational} and {\em negative} values for $\boldsymbol{r_e}[0]$. See Theorem \ref{thm:khalf_pks} and its proof in Section \ref{sec:khalf_complex}. Because this variant requires extra parameters, we only discuss its statement here informally.

In particular, the analysis depends on two condition numbers $\eta$ and $\xi$ (see Definitions \ref{def:eta} and \ref{def:xi}). These condition numbers are only required for the length ($\boldsymbol{r_e}[0]$). Intuitively, $\eta$ denotes the ratio between the magnitude of the most negative resource consumption for the length, and the smallest absolute value of the budget for that resource.\footnote{$\eta$ cares about $\min_{e \in E}\{\boldsymbol{r_e}[0]\}$, but not about $\max_{e \in E}\{\boldsymbol{r_e}[0]\}$. This is because we can safely ignore edges that have consumption bigger than the budget, but we cannot do so for edges (with negative lengths) that consume much less than the budget.} If the resource consumption for the length is always non-negative, then $\eta=0$. 
Similarly, $\xi$ denotes the ratio between the largest and the smallest absolute value of the budget for the length - intuitively, it is a measure of how useful the relaxation given to us by approximately fulfilling $\bdgt(s,t)[0]$ is. These condition numbers affect both the approximation ratio and the runtime of Theorem \ref{thm:khalf_pks}. The runtime of Theorem \ref{thm:khalf_pks} blows up by a  $\poly(\eta,\xi)$ factor compared to Theorem \ref{thm:vpks_khalf}, while the approximation factor of \ref{thm:khalf_pks} blows up by a $\polylog(\eta,\xi)$ factor.

\subsubsection{Resource-constrained junction trees (Main technical tool)} \label{sec:intro-junc}

In previous literature, one main engine for solving the directed pairwise spanner problem \cite{chlamtavc2020approximating,grigorescu2021online}, and the directed Steiner forest problem \cite{berman2013approximation,charikar1999approximation,feldman2012improved,chekuri2011set} is the notion of \emph{junction tree}. The notion of junction tree used for these problems is a subgraph $H$ rooted at a vertex $r$ and resolves a subset $D'$ of the demand pair set $D$. For our purpose, we extend the definition of junction trees to handle multiple resource constraints.

\begin{definition} [Resource-constrained junction tree] \label{def:vanilla_rcjt}
    Given an instance of \pks and a root vertex $r \in V$, a {\em resource-constrained junction tree} $\cJ$ rooted at $r$ is a non-empty subgraph in $G$ that has feasible $s\leadsto r\leadsto t$ walks for all $(s,t) \in D' \subseteq D$.
\end{definition}

The density of a junction tree is a ratio of its cost to the number of demand pairs resolved by that tree. We now formally define the problem of finding the {\em minimum density resource constrained junction tree}.

\begin{restatable}{definition}{defjuncvanilla} \label{def:junc_vanilla}
    \rcjt   
    
    \textbf{Instance}: Same as \pks.
    
    \textbf{Objective}: Find a root $r \in V$, a resource-constrained junction tree $\cJ$ rooted at $r$, such that the ratio  $\cost(\cJ)$ to the number of $(s,t)$ pairs that satisfy the resource requirement 
    \[\res(p(s,r,t)) := \sum_{ e \in p(s,r,t)} \boldsymbol{r_e} \preceq \bdgt(s,t)\] \footnote{ For two vectors $\boldsymbol{u}$ and $\boldsymbol{v}$ with the same dimension, we use $\boldsymbol{u} \preceq \boldsymbol{v}$ to denote that $\boldsymbol{u}[i] \le \boldsymbol{v}[i]$ for all elements $i$.}
    is minimized. Specifically, the goal is the following:
    \begin{equation}
             \min_{r \in V, \cJ}\frac{\cost(\cJ)}{|\{(s,t) \in D \mid \exists \: p(s,r,t) \text{ in } \cJ: \res(p(s,r,t)) \preceq \bdgt(s,t)\}| }.
    \end{equation}
\end{restatable}

We provide an approximation algorithm to find the \rcjt in Theorem \ref{thm:vpks_khalf}. 
This theorem is the main engine used to prove Theorem \ref{thm:vpks_khalf} and its proof appears in Section \ref{sec:jtree_vanilla_version}.

\begin{restatable}{theorem}{thmmdrcjtvanilla} \label{thm:mdrcjt_vanilla}

 (Integer version) If $\boldsymbol{r_e}[0] \in \{1, 2, \ldots, T\}$ for all $e\in E$,  then for any two constants $\ep>0, m>0$ there is  an $\tO(k^\ep \cdot \threshold^m)$-approximation algorithm  for \rcjt  that runs in $\poly(n, \threshold^m, T)$ time.
\end{restatable}

Theorem \ref{thm:mdrcjt_vanilla} generalizes the minimum density junction tree used for pairwise spanners \cite{chlamtavc2020approximating}, and for weighted spanners \cite{gkl2023}, with the same approximation ratio. We emphasize that our minimum-density junction tree framework accounts for resource constraints (as opposed to distance constraints). Furthermore, since junction trees and their variants are used in several algorithms in network design, we believe Theorem \ref{thm:mdrcjt_vanilla} which extends the junction tree black box in multiple ways are of independent interest. As an illustration, Theorem \ref{thm:rcs_keps} is a direct application of Theorem \ref{thm:mdrcjt_vanilla}. 

We also extend our results to the domain $\R$ for the length by defining and using a relaxed version of resource-constrained junction trees. This result - Theorem \ref{thm:mdrcjt_complex} - is in Section \ref{sec:jtree_complex_version}. This result also depends on the condition numbers $\xi$ and $\eta$ in the same way as Theorem \ref{thm:khalf_pks}. In Section \ref{sec:khalf_complex}, we use Theorem \ref{thm:mdrcjt_complex} to obtain Theorem \ref{thm:khalf_pks}. 

\subsubsection{Our results for \rcs}

We note that although our formulation of \pks greatly restricts the last $m$ resources, it is still a natural and useful tool. To illustrate the utility and versatility of our formulation of \pks, we establish reductions from \rcs to \pks.  Theorem \ref{thm:rcs_khalf} follows through a natural application of Theorem \ref{thm:vpks_khalf}. See Section \ref{sec:rcs} for details.

\begin{restatable}{theorem}{thmrcskhalf}
\label{thm:rcs_khalf}
When $ m, \eps > 0$ are constants, there exists a $\poly(n)$ time algorithm with an approximation factor of $\tO(k^{1/2+\ep})$ for \rcs.
\end{restatable}

We say $S' = \cup_{i \in \indexlist}S_i$ is essential if every demand pair $(s,t) \in D$ has to visit at least one vertex from $S'$. Theorem \ref{thm:rcs_keps} presents a specialized result of \rcs in terms of the size of the set $S'$. When $|S'|$ is small, the Theorem gives a very good approximation factor (especially when the essential subset is constant size).

\begin{restatable}{theorem}{thmrcskeps}
\label{thm:rcs_keps}
For a given instance of \rcs, let  $\indexlist \subseteq \{1,2,\ldots,m\}$. If $S' = \cup_{i \in \indexlist}S_i$ is essential, then when $m,\eps$ are constants, there exists a $\poly(n)$-time algorithm with an approximation factor of $\tO\left(|S'| \cdot k^{\ep}\right)$ for \rcs.
\end{restatable}

\subsubsection{Using \pks for constructing optimal hopsets}

A \emph{hopset} $H$ with hopbound $\beta$ and stretch $\alpha$ for a graph $G$ is a set of edges such that, 
for every pair of vertices $u, v \in V(G)$, there exists a walk in $G \cup H$ that uses at most $\beta$ hops 
and whose length is at most $\alpha$ times the actual distance between $u$ and $v$ in $G$. \cite{dinitz2025approximation} studies the optimization version of this problem, where the goal for a given graph instance is to determine the smallest set of edges that meets the hopbound $\beta$ and stretch $\alpha$ requirements. We show the versatility and the utility of the \pks formulation by recovering Theorem 5.1 from \cite{dinitz2025approximation} using Theorem \ref{thm:vpks_khalf}. See Section \ref{sec:gbh} for details.

\subsection{High-level technical overview}

We now sketch the proof of Theorem \ref{thm:mdrcjt_vanilla}. This is our main technical result and an important tool for deriving many of our other results, i.e., Theorems \ref{thm:vpks_khalf} and \ref{thm:rcs_keps}. Theorem \ref{thm:mdrcjt_complex} which further allows the length to be rational and negative needs one additional step which we describe in the end.

At a high level, junction trees form a cheap partial solution to connect a subset of the demand pairs. One possible solution for \pks is to iteratively select a low-density junction tree repeatedly until every demand pair is resolved. However, for \pks, we have to build \emph{feasible} junction trees which fulfill the resource constraints while connecting the demand pairs. We call this  modified junction-tree a {\em resource-constrained junction tree}. It is the main engine of our framework.

We start by describing how to  find a minimum-density resource-constrained junction tree rooted at a specific vertex $r$. Here the notion of \emph{density} is defined as the cost of the junction tree divided by the number of terminal pairs connected. This algorithm appeared in \cite{chlamtavc2020approximating} and is based on the framework for finding a minimum density junction tree (without distance constraints) given in \cite{chekuri2011set}. It is used for pairwise spanners with unit lengths and has three main steps:  1) use the input graph $G$ to build a layered graph. This layered graph captures the distance constraints (i.e., there is only one resource) of the original problem using connectivity constraints; 2) use a \emph{height reduction} technique from \cite{zelikovsky1997series} to construct a shallow tree-like graph from the layered graph, by paying a small approximation ratio;  and finally 3) solve the corresponding \emph{minimum-density Steiner label cover} problem on the shallow graph. The minimum-density Steiner label cover instance is solved by building, solving, and rounding a linear program. The blackbox used to round this linear program has a good integrality gap when the graph is shallow and it is for this reason, we need step 2. 
To extend the previous results to handle multiple resources, we develop new ideas which significantly depart from the work of \cite{chlamtavc2020approximating}. More importantly, allowing multiple resources also destroys some structural properties that were essential for the techniques in \cite{chlamtavc2020approximating}. We make a number of careful observations and some natural assumptions to work around these structural deficiencies. 
We highlight the main strategies used to build our minimum density junction tree framework in the next  three paragraphs. In the fourth paragraph we describe how to extend the framework to allow the consumption of the length to be rational. Finally, in the fifth paragraph, we describe how to design an approximation algorithm based on the minimum density junction tree framework.

\paragraph{Turning resource constraints into connectivity constraints.}

In our first step, we construct a \emph{product graph} that captures the resource constraints using connectivity constraints. This product graph comprises multiple copies of each vertex $u$ in the original graph $G$. These copies are of the form $(u,\boldsymbol{I})$ (where $\boldsymbol{I}$ is an integer vector with $m+1$ elements). The product graph is structured in such a way that we can find a walk from $(u,\boldsymbol{I})$ to the root $r$ if and only if there is a walk of resource consumption $\boldsymbol{I}$ from the vertex $u$ to the root $r$. The copies we create and the edges we use to connect them need to handle multiple resources unlike the work of \cite{chlamtavc2020approximating}. To handle negative edge lengths, we use the key modification introduced by \cite{grigorescu2024directed} which allows edges to go \emph{backward}, instead of always forward as in \cite{chlamtavc2020approximating}.

A major structural challenge arises in this step since we are dealing with  more than one resource. When $m=0$ (i.e., in weighted and unit-cost spanners), if we have multiple demand pairs associated with the same vertex $s$, we just need one $s \leadsto r$ walk - this walk just needs to be short enough for the strictest distance constraint associated with $s$ (i.e., when $m=0$ all resource constraints associated with a specific terminal $s$ are dominated by one constraint since a set of one-dimensional vectors form a total ordering). But when $m \geq 1$, there is no notion of the strictest constraint. For instance, if $m=1$, the vertex $s$ could be a part of two demand pairs $(s,t_1)$ and $(s,t_2)$ with a resource budget of $(10,20)$ and $(20,10)$ (since a set of two-dimensional vectors form a partial ordering). One would need two separate $s \leadsto r$ walks in this example and potentially as many as $|D|$ (recall that $D$ is the set of demand pairs for \pks) $s \leadsto r$ walks in a more general example which makes the solution prohibitively expensive in comparison to the optimal solution. This implies that a collection of $s \leadsto r$ walks no longer has an in-arborescence structure, unlike traditional junction trees.

To resolve this issue, we show a key observation - the over counting of the edges can instead be upper bounded by the number of layers that correspond to any $m$ of the $m+1$ resources. This means that we can still obtain a good approximation if we bound the budget for $m$ of the $m+1$ resources. We use this observation to allow greater precision and freedom for the length. This key observation is a crucial step in obtaining a good approximation factor for \pks under natural assumptions which allow us to solve \rcs and other interesting problems. 

A secondary structural challenge is that for a covering resource $i$, a walk $p$ with $\res(p)[i] \leq -\tau \leq \bdgt[i]$ may also be a feasible walk that is used to connect a demand pair. To directly track such a walk with a product graph, we need to build an unlimited number of layers for the $i^{th}$ resource which in turn makes the approximation factor unbounded. We get around this challenge by making some careful observations on the nature of covering resources which allow us to get away with just $\tau$ layers for any covering resource. This is explained in detail in Section \ref{sec:jtree_vanilla_version}.

Finally, for every demand pair $(s,t) \in D$, we add a number of dummy terminals of the form $(s^t,\boldsymbol{I})$ and $(t^s,\boldsymbol{J})$. These dummy terminals are structured in such a way that there is a walk from $(s^t,\boldsymbol{I})$ to $(t^s,\boldsymbol{J})$ in the product graph if and only if in the original graph $G$ there is a $s \leadsto t$ walk whose resource consumption is at most $\boldsymbol{J}+\boldsymbol{I}$. Now the resource constraints can be modeled by connectivity constraints involving the correct dummy terminals.

\paragraph{Turning the product graph into a shallow tree}

In the second step, we apply the height reduction result from \cite{zelikovsky1997series} to turn the product graph into a shallow graph. Now, to solve the resource constraint problem we originally had, we just need to solve a variant of the Steiner problem, namely, \emph{minimum density Steiner label cover} on a shallow graph. A crucial property of the black box from \cite{zelikovsky1997series} is that we can still keep track of the dummy terminals after obtaining the shallow graph and that allows us to enforce our constraints.

\paragraph{Solving the \emph{minimum density Steiner label cover} instance}

We first build a standard linear program \cite{chlamtavc2020approximating,chekuri2011set} for the given instance of \emph{minimum density Steiner label cover} and obtain its fractional solution. Before we round this fractional solution, we need to ensure that for any $(s,t) \in D$, whatever copies of the form $(s^t,\boldsymbol{I})$ and $(t^s,\boldsymbol{J})$ we select will still fulfill the resource constraint. Note that a specific copy $(s^t,\boldsymbol{I_1})$ might fulfill the resource constraint when paired with $(t^s,\boldsymbol{J_1})$ but may not do so when paired with $(t^s,\boldsymbol{J_2})$ i.e., it might be the case that $\boldsymbol{I_1} + \boldsymbol{J_1} \preceq \bdgt(s,t)$ but $\boldsymbol{I_1} + \boldsymbol{J_2} \succ \bdgt(s,t)$. The linear program may still assign some value for the dummy terminal $(s^t,\boldsymbol{I_1})$ since it can be paired with $(t^s,\boldsymbol{J_1})$. The same might happen for  $(t^s,\boldsymbol{J_2})$ since it might be the case that a different $(s^t,\boldsymbol{I_2})$ can be paired with it instead. And then if we round the fractional solution of the linear program, we might end up rounding $(s^t,\boldsymbol{I_1})$ and $(t^s,\boldsymbol{J_2})$. But a walk from $(s^t,\boldsymbol{I_1})$ to $(t^s,\boldsymbol{J_2})$ will not fulfill the resource constraint.

Thus, for a specific demand pair $(s,t)\in D$, we need to pick a subset of the dummy terminals so that no matter what dummy terminal we round, the resource constraint for that demand pair would be satisfied. We say that a dummy terminal $(s^t,\boldsymbol{I_1})$ is {\em bad} for a resource $i$ if there is a different dummy terminal $(t^s,\boldsymbol{J_1})$ such that $\boldsymbol{I_1}[i]+\boldsymbol{J_1}[i] > \bdgt(s,t)[i]$; we say it is {\em good} otherwise. A pruning process to remove the bad dummy terminals was introduced in \cite{chlamtavc2020approximating} when there is only one resource. 

However, our specific formulation has a new challenge.  What happens if a dummy terminal $(s^t,\boldsymbol{I})$ is a bad terminal for the length but not for the second? How do we ensure that pruning the bad terminals for resource $i$ does not delete all the good terminals for a different resource $j$? Our work takes advantage of some structural properties of the product graph and delicately adapts the pruning process to the multiple resource setting. Our pruning process occurs in phases where in each phase we deal with the constraints associated with one specific resource. We show that any phase of the new pruning process will only delete at most half the good terminals associated with the other resources. There are  $m+1$ phases, each corresponds to one resource. The number of unpruned terminals associated with any resource will only drop by a factor of $2^{m+1}$. Therefore, we can prune the bad terminals and then scale the unpruned terminals with factor of $2^{m+1}$ when rounding them. 

After the pruning process, we just need to round the solution of the linear program using Group Steiner rounding and that gives us our resource-constrained junction tree.

\paragraph{Scaling and rounding the edge lengths} This step is only required for Theorem \ref{thm:mdrcjt_complex} i.e., when we have rational edge lengths. For Theorem \ref{thm:mdrcjt_complex}, we run this step before any of the other three steps. In this step, we use an approach similar to \cite{horvath2018multi,grigorescu2024directed} to properly scale and round the length. This process effectively turns the rational values into integers (which are smaller when the problem is well conditioned). Note that the last $m$ resources are already integers and therefore require no scaling. We round up the length and for this reason, the consumption of the length might be slightly overestimated. This rounding and scaling process allows us to construct a product graph of size polynomial in the condition numbers. This product graph approximately can model the resource constraints using connectivity constraints. See Figure \ref{fig:original_scaled_graph} for an example of the scaling process. This step would not be required when the length is restricted to integers polynomial in $n$ i.e., for Theorem \ref{thm:mdrcjt_vanilla}.

\paragraph{The approximation algorithm for \pks} 

  The approximation algorithm for the integer version of \pks i.e., Theorem \ref{thm:vpks_khalf} follows a standard iterative density procedure used for directed Steiner forests \cite{chekuri2011set} and spanners \cite{grigorescu2021online,gkl2023}. Iteratively picking minimum-density distance-constrained junction trees only pays a factor of $O(\sqrt{k})$. Combining this and Theorem \ref{thm:mdrcjt_vanilla} results an $\tO(k^{1/2 + \ep} \tau^m)$-approximation.

\subsection{Future directions}
A very interesting follow-up question would be  if one can get a version of Theorem \ref{thm:mdrcjt_vanilla} where the approximation factor does not depend on $\tau$. Even making the dependence sublinear in $\tau$ would be of great utility for constructing hopsets and many other applications. Another interesting follow-up problem would be to allow $\boldsymbol{r_e}[i]$ to be both positive and negative for the same $i$. This extension would allow us to address an entirely new line of applications. For instance, allowing this would let us track a resource like fuel which can both be gained (by going to gas station) and lost (by travelling) in a long trip. A much simpler question would be to see what other applications one can find for \pks. We believe that further research along these lines would make the connection between network design in theory and practice stronger. 

\subsection{Additional Related work}

A well-studied variant of spanners is called the \emph{directed $s$-spanner} problem, where 
there is a fixed value $s \geq 1$ called the \emph{stretch}, and the goal is 
to find a subgraph with a minimum number of edges such that the distance between \emph{every} pair of
vertices is preserved  within a factor of $s$ in the original
graph. The \emph{directed $s$-spanner} problem is studied in \cite{elkin1999client, Kortsarz2001OnTH,berman2013approximation,dinitz2016approximating}. The problem is hard to approximate within an $O(2^{{\log^{1-\eps} n}})$ factor for $3 \leq s = O(n^{1-\delta})$ and any $\eps, \delta \in (0,1)$, unless $NP\subseteq  DTIME(n^{\operatorname{polylog} n})$ \cite{ElkinP07}. More general variants of spanners are studied in \cite{chlamtavc2020approximating,elkin1999client,bhattacharyya2012transitive}.

Recently, \cite{gkl2023} studied the {\em weighted} spanner problem for arbitrary terminal pairs, which has a closer formulation to multi criteria spanners. \cite{grigorescu2024directed} studied an extension of {\em weighted} spanners called buy-at-bulk spanners. \cite{grigorescu2021online} studied online directed spanners. We refer the reader to 
the excellent survey \cite{ahmed2020graph} for a more comprehensive exposition. 

In the {\em group steiner metric}, we are given a collection of groups of must-visit vertices, and we measure the distance between any two vertices as the length of the shortest path between them that
traverses at least one must-visit vertex from each group. \cite{bilo2024graph} presents an algorithm to build a spanner w.r.t. the {\em group steiner metric} for the undirected setting. We note that \rcs is a generalization of this metric in the directed setting. \rcs allows edge costs to arbitrary (\cite{bilo2024graph} considers edge costs to be uniform). Our model is for the pairwise setting and we allow individual pairs to skip some groups (\cite{bilo2024graph} is for the all pair setting and assumes that every vertex pair needs to visit every group). More importantly, our model also allows us to specify vertices to be avoided for each demand pair $(s,t) \in D$. 

\subsection{Organization}

We first present the integer version of our results - Theorem \ref{thm:mdrcjt_vanilla} (our integer version result for \rcjt) and Theorem \ref{thm:vpks_khalf} (our integer version approximation algorithm for \pks) in Section \ref{sec:jtree_vanilla_version} and Section \ref{sec:khalf_vanilla} respectively. We present our results for \rcs and \gbh in Section \ref{sec:rcs} and \ref{sec:gbh} respectively. Finally, we present the rational version of our results - Theorem \ref{thm:mdrcjt_complex} and Theorem \ref{thm:khalf_pks} in Section \ref{sec:jtree_complex_version} and Section  \ref{sec:khalf_complex} respectively. We also give a compilation of the notation in Section \ref{sec:appendix}.

\section{Proof of theorem \ref{thm:mdrcjt_vanilla}} 
\label{sec:jtree_vanilla_version} 

The overall structure of this proof is based on \cite{chlamtavc2020approximating,grigorescu2024directed} which are in turn based on \cite{chekuri2011set}. We do the following procedure using every possible root $r \in V$ and then take whichever case has the minimum density among all the possible roots. First, we build a product graph to turn the resource constraints/budgets into connectivity constraints. This gives rise to an instance of the \mslc problem. Then we use the height reduction result from  \cite{zelikovsky1997series} to obtain a shallow tree  for the instance of \mslc. Finally, we can solve this instance by building an LP, solving it and using the \gsf problem to round it using an extension of the approach in \cite{chlamtavc2020approximating}.

In this procedure we have to deal with resource constraints as opposed to \cite{chlamtavc2020approximating} which deals with distance constraints. This departure introduces new technical challenges as turning the resource constraints into connectivity constraints is more complex than turning distance constraints to connectivity constraints. Dealing with cases where $m \geq 1$, also introduces a new structural challenge. When $m=0$ (i.e., in weighted and unit-cost spanners), if we have multiple demand pairs associated with the same vertex $s$, we just need one $s \leadsto r$ walk - this walk just needs to be short enough for the strictest distance constraint (i.e., when $m=0$ the resource constraints have a dominating set of size 1). But when $m \geq 1$, there is no notion of the strictest constraint (i.e., your dominating set is potentially of infinite size). For instance, if $m=1$, the vertex $s$ could be a part of two demand pairs $(s,t_1)$ and $(s,t_2)$ with a resource budget of $(10,20)$ and $(20,10)$. One would need two separate $s \leadsto r$ walks in this example and potentially as many as $|D|$ (recall that $D$ is the set of demand pairs for \pks) $s \leadsto r$ walks in a more general example which makes the solution prohibitively expensive in comparison to the optimal solution. For this reason, we need strong assumptions on the resource consumptions (for the last $m$ resources) and a deeper analysis than \cite{chlamtavc2020approximating,chekuri2011set}

In addition, solving the \mslc problem after using height reduction also has additional complications when compared to \cite{chlamtavc2020approximating}. The pruning procedure in \cite{chlamtavc2020approximating} that is used to carefully process the \mslc problem to get a good input for \gsf needs to be extended to the case where $m \geq 1$.

\subsection{Turning distance constraints into connectivity constraints}

\paragraph{High-level idea and potential challenges:}

For a specific root vertex $r$, we turn our distance constraints with edge costs problem into a connectivity problem with edge costs. The overall process for this is as follows. We first build a product graph. After building this product graph, a number of dummy terminals are added to various layers - now the problem of finding a $s \leadsto r \leadsto t$ walk within the resource budget changes into a problem of connecting the correct dummy terminals. 

While the overall approach of turning a distance constraint into a connectivity constraint was introduced in \cite{chlamtavc2020approximating}, we have to keep several things in mind while designing a similar construction. We are dealing with multiple resources as opposed to just one resource (i.e., length). This means that our product graph needs to be more complex than that of \cite{chlamtavc2020approximating}. As mentioned earlier, our analysis also needs to be more detailed due to structural changes caused by having more than one resource.

\paragraph{Graph construction:} Let us now see our graph construction which needs to keep all of these concerns in mind. Let the vectors $\boldsymbol{t^{-}}$ and $\boldsymbol{t^{+}}$ represent the smallest and largest possible value the resource consumption of a feasible walk could take. Let $\zvec$ be a vector of $m$ dimensions with zeroes everywhere. Observe that $\boldsymbol{t^-}[0] = 0$ and $\boldsymbol{t^+}[0] = \bdgt_{max}[0] \leq O(T \cdot \poly(n))$.  

For $i \in [1,m]$, 

\begin{itemize}
    \item when $\boldsymbol{r_{e}}[i] \geq 0 \:\: \forall e \in E$ i.e., $i$ is a packing resource, $\boldsymbol{t}^{-}[i] = 0$ (since all resource consumptions for this resource are non negative, the total resource consumption for this resource is at least zero) and $\boldsymbol{t}^{+}[i] = \threshold$.
    \item when $\boldsymbol{r_{e}}[i] \leq 0\:\: \forall e \in E$ i.e., $i$ is a covering resource, $\boldsymbol{t}^{+}[i] = 0$ and $\boldsymbol{t}^{-}[i] = -\threshold$.
\end{itemize}. 

We construct a product graph $\bar{G}_r$ with the following vertices:

\begin{align}
    & \Bar{V}_r^L = \left((V\setminus r) \times \prod_{i=0}^{m} \{\boldsymbol{t^{-}}[i],\ldots,-2 ,-1 ,0 ,1 ,2,\ldots,\boldsymbol{t^{+}}[i] \}  \times \{L\}\right) \cup \{(r,\zvec,L)\}, \\
    & \Bar{V}_r^R = \left((V\setminus r) \times \prod_{i=0}^{m} \{\boldsymbol{t^{-}}[i],\ldots,-2 ,-1 ,0 ,1 ,2,\ldots,\boldsymbol{t^{+}}[i] \} \times \{R\}\right) \cup \{(r,\zvec,R)\}, \\
    & \bar{V}_r = \bar{V}_r^R \cup \bar{V}_r^L.
\end{align}

As an example, a vertex in the newly constructed graph looks as follows: $(u,\boldsymbol{I}, L)$. This denotes that the new vertex is a copy of the vertex $u$ from the graph $G$, and if this copy can reach the root in the product graph $\bar{G}_r$, then the vertex $u$ can reach the root with a resource consumption of  $\boldsymbol{I}$ in the original graph $G$. Let us call this $\boldsymbol{I}$ (note that $\boldsymbol{I} \in Z ^ {m+1}$) as the label of the layer. $L$ (left) and $R$ (right) are labels that are used to distinguish two separate copies of the same vertex set. Vertices of the form $(u,\boldsymbol{I}, L)$ cannot be reached from the root; on the other hand, vertices of the form $(u,\boldsymbol{I}, R)$ cannot reach the root.

We connect these vertices with the following edges:

\begin{equation}
    \begin{split}
        \Bar{E}_r^R &= \{((u,\boldsymbol{I},R)(v,\boldsymbol{J},R)) \mid (u,\boldsymbol{I},R),(v,\boldsymbol{J},R) \in \bar{V}_r^R,(u,v)\in E \text{ and } \boldsymbol{\Bar{r}}_{(u,v)}=(\boldsymbol{J}-\boldsymbol{I} ) \text{ where } \boldsymbol{I},\boldsymbol{J} \in  Z^{m+1} \\ &\quad \text{ and } \boldsymbol{\Bar{r}}_{(u,v)} \text{ is the resource consumption of the edge} (u,v)\}.
    \end{split}
    \notag
\end{equation}

\begin{equation}
    \begin{split}
        \Bar{E}_r^L &= \{((u,\boldsymbol{I},L)(v,\boldsymbol{J},L)) | (u,\boldsymbol{I},L),(v,\boldsymbol{J},L) \in \bar{V}_r^L,(u,v)\in E
        \text{ and } \boldsymbol{\Bar{r}}_{(u,v)}=(\boldsymbol{I} - \boldsymbol{J} ) \text{ where } \boldsymbol{I},\boldsymbol{J} \in  Z^{m+1} \\
        & \quad \text{ and } \boldsymbol{\Bar{r}}_{(u,v)} \text{ is the resource consumption of the edge} (u,v)\}.
    \end{split}
\end{equation}

Intuitively, we add an edge between two vertices whenever it makes sense i.e., when the original copies of these two vertices are connected in the original graph $G$ and the layer separation between these two vertices is equal to the resource consumption of the edge in $G$. The edges in our product graph $\Bar{G}_r$ inherit the costs from the corresponding edges in the original graph. 

Let the resource consumption of a walk $P$ be denoted by $\res(P)$. The edges and vertices are built in such a way that if there is a walk from $(u,\boldsymbol{I},L)$ to the root $r$, then in the graph there is a $u \leadsto r$ walk $p(u,r)$ such that $\res(p(u,r)) = \boldsymbol{I}$. Similarly, if there is a walk from the root $r$ to $(u,\boldsymbol{I},R)$, then in the graph there is a $r \leadsto u$ walk $p(r,u)$ such that $\res(p(r,u)) = \boldsymbol{I}$.

Note that $(\bar{G}_r^R = (\bar{V}_r^R,\bar{E}_r^R))$ and $(\bar{G}_r^L = (\bar{V}_r^L,\bar{E}_r^L))$ are both disjoint until now. We add one final edge. This connects $(r,\zvec,L)$ to $(r,\zvec,R)$ and thus connects $\bar{G}_r^L$ to $\bar{G}_r^R$. It is a dummy edge and therefore has zero cost.

Let

\begin{equation}
    \Bar{E}_r = \Bar{E}_r^R \cup \Bar{E}_r^L.
\end{equation}

\begin{equation}
    \Bar{E}_r = \Bar{E}_r \cup \{((r,\zvec,L),(r,\zvec,R))\}.
\end{equation}

\begin{equation}
    \bar{G}_r = (\bar{V}_r,\bar{E}_r).
\end{equation}

\begin{definition} \label{de:valid_layer}
    We call a vector $\boldsymbol{I} \in Z^{m+1}$ {\em valid} if $\boldsymbol{t^{-}} \preceq \boldsymbol{I} \preceq \boldsymbol{t^{+}}$.
\end{definition} 

For every terminal pair $(s,t) \in D$, do the following, 

\begin{enumerate}
    \item Add new vertices $(s^t,\boldsymbol{I})$ and $(t^s,\boldsymbol{J})$ for all {\em valid} vectors $\boldsymbol{I},\boldsymbol{J}$  to $\bar{V}_r$.
    \item For all such $\boldsymbol{I}$ and $\boldsymbol{J}$ add edges $((s^t,\boldsymbol{I})(s,(\boldsymbol{I},L)))$ and $((t,(\boldsymbol{J},R))(t^s,\boldsymbol{J}))$ with zero cost to $E_r$.
    \item Now for every terminal pair $(s,t) \in D$ define:
    \begin{enumerate}
        \item terminal sets $S_{s,t} = \{(s^t,\boldsymbol{I}) \:\forall \text{ {\em valid} vectors }\boldsymbol{I}\}$,
        \item $T_{s,t} = \{(t^s,\boldsymbol{J}) | \:\forall \text{ {\em valid} vectors }\boldsymbol{J}\}$ and 
        \item relation $R_{s,t} = \{(s^t,\boldsymbol{I}),(t^s,\boldsymbol{J}) \in S_{s,t}\times T_{s,t} | (\boldsymbol{I}+\boldsymbol{J}) \preceq \bdgt(s,t) \}$.
    \end{enumerate}
\end{enumerate}

However, the above construction does not handle one specific case. Recall that for $i \in [p,p+c]$, $t^-[i] = -\tau$ for covering resources. However, for a given demand pair $(s_j,t_j)\in D$, a walk $p$ with $\res(p)[i] \leq -\tau \leq  \bdgt[i]$ may also be a feasible walk that is used to connect a demand pair. The above construction cannot track any walk with $\res(p)[i]\leq -\tau$. 

Observe that, since a covering resource always has $\boldsymbol{r_e}[i] \leq 0$, once a walk $p$ has $\res(p)[i] \leq \bdgt[i]$, adding more edges to the walk will never make $\res(p)[i] > \bdgt[i]$. So, if a walk has $\res(p)[i] \leq \bdgt[i]$, we don't need to track $\res(p)[i]$ any more. We just need to ensure that the walk leads to the root i.e., at this stage, we only need to preserve connectivity not the resource constraint. Thus, we don't need to build any extra layers beyond $\bdgt[i] >= t^-[i]$ for the $i^{th}$ resource. To remedy the issue mentioned in the previous paragraph, we add another set of edges. For every covering resource $i$, every valid vector $\hat{I}$ of the form $(\ldots,t^-[i],\ldots)$ and for every edge $(u,v) \in E$, we add an edge from $(u,\hat{I})$ to $(v,\hat{I})$ to $\bar{E_r}$. These edges will allow us to track walks with $\res(p)[i] \leq -\tau \leq \bdgt[i]$. Note that these new edges will not introduce any new feasible walk that does not already exist in the original graph. This is because any subwalk $p'$ with $\res(p')[i] \leq -\tau$ resource consumption cannot be combined with more edges to create a walk $p$ with $\res(p')[i] > -\tau$ since $\boldsymbol{r_e}[i] \leq 0$ for every edge when $i$ is a covering resource.

\paragraph{Relating the product graph with the original graph:}

Let us now create a simpler graph: Let $\hat{G}$ be a graph comprised of two copies of $G$ - namely $G_{-}$ and $G_{+}$ intersecting in the node $r$. For every vertex $u \in V$, let $u_{+}$ and $u_{-}$ denote the copies of $u$ in $G_{+}$ and $G_{-}$ respectively. We call this graph as the {\em intersection graph}. For convenience, we use $\validlayer$ to denote the following expression: $\validlayer = \prod_{i=1}^{m} \left(O(\boldsymbol{t^{-}}[i]+ \boldsymbol{t^{+}}[i] \right))$.

The following two claims will relate the product graph $\bar{G_r}$ with the intersection graph $\hat{G}$ (and thus indirectly the original graph $G$).

\begin{claim} \label{cl:junctree_layeredoriginalcomparison}
    For any $f > 0$, and set of terminal pairs $D' \subseteq D$, assume there exists a subgraph $\hat{H}$ in $\hat{G}$ of total cost $\leq f$ containing a walk of resource consumption at most $\bdgt(s,t)$ from $s_+$ to $t_-$ for every $(s,t) \in D'$. Then there exists a subgraph $\bar{H_r}$ in $\bar{G_r}$ of total cost $\leq f \cdot \validlayer$ containing a walk from $(s^t, \boldsymbol{I})$ to $(t^s, \boldsymbol{J})$ such that $((s^t,\boldsymbol{I}),(t^s,\boldsymbol{J})) \in R_{s,t}$ for every $(s,t) \in D'$. 

    \end{claim}
    
\begin{proof}
     Let $p(s_i,t_i)$ be a $s_i \leadsto t_i$ walk in $\hat{H}$. Due to the construction of $\hat{G}$, this walk is forced to go through the root $r$ and $\res(p(s_i,t_i) \leq \bdgt(s,t)$. Let $p(s_i,r)$ be a subwalk of $p(s_i,t_i)$ that ends at $r$ and let $\res(p(s_i,r)) = \boldsymbol{I}$. By construction $\boldsymbol{t^-} \preceq \boldsymbol{I} \preceq \boldsymbol{t^+}$. Similarly, let $p(r,t_i)$ be a $r \leadsto t_i$ subwalk of $p(s_i,t_i)$ and let $\res(p(s_i,r)) = \boldsymbol{J}$. By construction $\boldsymbol{t^-} \preceq \boldsymbol{J} \preceq \boldsymbol{t^+}$. 
        
    Observe that $\boldsymbol{J} + \boldsymbol{I} \preceq \bdgt(s,t)$. Additionally,  by construction $(s^t,\boldsymbol{I}),(t^s,\boldsymbol{J}) \in \bar{V}_r$ and thus $((s^t,\boldsymbol{I}),(t^s,\boldsymbol{J})) \in R_{s,t}$. Note that $(s^t,\boldsymbol{I})$ and $(t^s,\boldsymbol{J})$ can be connected by a walk analogous to $p(s_i,t_i)$. Compile all such walks and we will have $\bar{H_r}$.

    For the cost note that we effectively use the same set of edges but since $\bar{G_r}$ has multiple copies of the edges from $G$, it could over count some edges. This issue (which does not arise in \cite{chlamtavc2020approximating}) is why we need to be careful about how many layers we build. In \cite{chlamtavc2020approximating}, we have $m=0$ and therefore the layers we include have a dominating set of size $1$ (i.e., when we are trying to reach the root, there is no need to consider an edge at a distance of $i$ from the root when we include a different copy of the same edge at a distance of $j$ from the root with $j < i$). That is not the case when $m \geq 1$. 
    
    For any layer $\boldsymbol{I}$, we call the assignment of the last $m$ elements of $\boldsymbol{I}$ as its configuration. We now note that for a specific configuration, we need only one copy of an edge. This is because when we fix a configuration, a specific assignment of the first resource will dominate all other assignments (and thus we have a dominating set of size $1$ after we fix a configuration). In total, the number of configurations we have is given by $\validlayer$ and that proves the rest of the claim.

\end{proof}

\begin{claim}

    For any $f > 0$, assume there exists a subgraph $\bar{H_r}$ in $\bar{G_r}$ of total cost $\leq f$ containing a walk from from $(s^t, \boldsymbol{I} )$ to $(t^s, \boldsymbol{J} )$ such that $((s^t,\boldsymbol{I}),(t^s,\boldsymbol{J})) \in R_{s,t}$ for every $(s,t) \in D'$. Then there exists a subgraph $\hat{H}$ in $\hat{G}$ of total cost $\leq f$ containing a walk of resource consumption at most $\bdgt(s,t)$ from $s_+$ to $t_-$ for every $(s,t) \in D'$.
\end{claim}

\begin{proof}
    This claim is much simpler. Let $p(s_i,t_i)$ be a $s_i \leadsto t_i$ walk in $\bar{H_r}$. Due to the construction of $\bar{G_r}$, this walk is forced to go through the root $r$ and $\res(p(s_i,t_i) \leq \bdgt(s,t)$. Due to the construction of $\bar{G_r}$, we can see that there is a $s_i \leadsto r \leadsto t_i$ walk $p'$ in $\hat{G}$ with the same cost and resource consumption $\preceq \res(p(s_i,t_i)) \preceq \bdgt(s,t)$. 

    Compile all such walks and we have $\hat{H}$. The overall cost of this set $\leq f$ since we can reuse the same set of edges (in practice the cost can be lower since $\hat{G}$ has only two duplicate copies of any edge unlike $\bar{G_r}$ which has several duplicate copies).

\end{proof}

\paragraph{Runtime:} 
Our runtime so far is polynomial in $|\bar{G}_r|$ (for any graph $G$, we use $|G|$ to denote the number of vertices in $G$). Note that $|\bar{G}_r|$ is a polynomial in $O(n \cdot \Pi_{i=1}^{i=m+1}(|\boldsymbol{t^{-}}[i]| + |\boldsymbol{t^{+}}[i]| ))$.

\paragraph{To summarize,} we have turned all budget constraints into connectivity constraints so far. We keep track of the budget constraints using some relations. 

\subsection{Using Height reduction}

We now state the height reduction lemma from \cite{cekp} here which is in turn a modification of the height reduction result from \cite{helvig2001improved}.

\begin{restatable}{lemma}{lemmahrbb}
    \label{le:height_reduction_bb} \cite{cekp}
    (\textbf{Height Reduction}) Given a directed graph $G=(V,E)$ with edge costs $c(e)$, for all $h>0$, we can efficiently find a downward directed, layered graph $G_r^{down}$ on $(h+1)$ levels and edges (with new edge costs) only between consecutive levels going from top (level $0$) to bottom (level $h$), such that each layer has $n$ vertices corresponding to the vertices of $G$, and, for any set of terminals $S$ and any root vertex $r$,
    \begin{itemize}
        \item  the optimal objective value of the single-source problem to connect $r$ (at level $0$) with $h$ (at level $h$) on the graph $G^{down}_r$ is at most $O(hk^{1/h})\rho$, where $\rho$ is the objective value of an optimal solution of the same instance on the original graph G,
        
        \item given an integral (fractional solution) of objective value $\rho$ for the single-sink problem to connect $r$ with $S$ on the graph $G^{down}_r$, we can efficiently recover an integral (fractional solution) of objective value at most $\rho$ for the problem on the original graph $G$.
    \end{itemize}
    In the same way, we can obtain a upward directed, layered graph $G^{up}_r$ on $(h+1)$-levels with edges going from bottom to top, satisfying the same properties as above except for single-sink (as opposed to single-source) instances instead.
\end{restatable}

We now apply Lemma \ref{le:height_reduction_bb} on $(\bar{G}_r^R = (\bar{V}_r^R,\bar{E}_r^R))$ and $(\bar{G}_r^L = (\bar{V}_r^L,\bar{E}_r^L))$ separately, and obtain two new $(h+1)$ layered graphs $\bar{G}^{up}_r$ and $\bar{G}^{down}_r$ where $h$ is some positive integer which depends on $1/\eps$. We create a new graph $\bar{T_r}$ by adding an edge from the root $r^{up}$ in $\bar{G}^{up}_r$ to the root $r^{down}$ in $\bar{G}^{down}_r$. We can only keep track of the terminal vertices after using Lemma \ref{le:height_reduction_bb} and that is sufficient for our needs. Let $\psi: V(\bar{T_r}) \mapsto V(\bar{G}_r)$ be a mapping that relates terminals in the height reduced graph $\bar{T_r}$ to the terminals in the product graph $\bar{G}_r$.

Let $w_r$ be the cost on the graph $\bar{T}_r$. Let $S_{s,t}^r = \{\psi^{-1}((s^t,\boldsymbol{I})) \: \forall (s^t,\boldsymbol{I}) \in S_{s,t}\}$ contain all terminals in $\bar{T}_r$ that correspond to the terminals in $S_{s,t}$. Similarly, let $T_{s,t}^r = \{\psi^{-1}((t^s,\boldsymbol{I})) \: \forall (t^s,\boldsymbol{I}) \in T_{s,t}\}$ contain all terminals in $\bar{T}_r$ that correspond to the terminals in $T_{s,t}$ . Finally, set $R_{s,t}^r = \{(s^t_r,\boldsymbol{I}),(t^s_r,\boldsymbol{J}) \in S_{s,t}^r\times T_{s,t}^r \mid \boldsymbol{I}  +\boldsymbol{J}\preceq \bdgt(s,t) \} $.

We now define the \mslc problem.

\begin{definition}
    In the \mslc problem, we have a directed graph $G = (V,E),$ two collections of disjoint vertex sets $\hat{S},\hat{T} \subseteq 2^V$, a collection of set pairs $P\subseteq \hat{S} \times \hat{T},$ and for each set pair $(S,T) \in P,$ a relation $R(S,T) \subseteq S \times T$ and non-negative edge costs $c: E \to \R_{\geq 0},$. The objective here is to find an edge set $F\subseteq E$ that minimizes the ratio
    \begin{equation*}
        \frac{\sum_{e \in F}c(e)}{|\{(S,T)\in P \mid \exists (s,t) \in 
 R(S,T):\text{ $F$ contains an $s \leadsto t$ walk }\}|}.
    \end{equation*}
\end{definition}

Now, to prove Theorem\ref{thm:mdrcjt_vanilla}, we just need to show that we can achieve an $O(n^\ep)$ approximation for the \mslc instance $(\bar{T}_r,\{S_{s,t}^r, T_{s,t}^r, R_{s,t}^r \mid (s,t) \in D\},w_r)$ obtained from our reduction. In other words, it suffices to show the following lemma.

\begin{lemma} \label{le:junction_tree_intermediate}
    In the given setting, there exists a $O\left(\log^3 |\bar{G_r}|\right) = O\left(\log^3 \left(n \cdot \Pi_{i=0}^{i=m}\left(\boldsymbol{t^{-}}[i] + \boldsymbol{t^{+}}[i] \right)\right)\right)$ approximation algorithm for the following problem that runs in time polynomial in $|\bar{G_r}|$.

    \begin{itemize}
        \item Find an edge set $F \subseteq T_r$ minimizing the ratio
        \begin{equation*}
            \frac{\sum_{e \in F}w_r(e)}{\mid\{(s,t)\in P \mid \exists (\hat{s},\hat{t}) \in 
 R_{s,t}^r:\text{ F contains an $\hat{s}$-$\hat{t}$ walk }\}|}.
        \end{equation*}
    \end{itemize}
\end{lemma}

\subsection{Proof of Lemma \ref{le:junction_tree_intermediate} - Reduction to \gsf}
\paragraph{Preliminary LP formulation:}
We follow a structure similar to that of \cite{chlamtavc2020approximating} here for the proof of \ref{le:junction_tree_intermediate}. The only change we need to make is in the relation $R_{s,t}^r$ that we are using. This makes it necessary for us to make some careful and very important changes in the overall proof as we shall see soon. The following LP is identical to the LP in \cite{chlamtavc2020approximating} (with slightly different notation) with the sole exception of the relation. Note that the objective in the following LP tries to minimize the density of the solution obtained which is what we need for the junction tree. 

\begin{subequations} \label{lp:junc_flow}
\begin{align} 
& \min & & \sum_{e \in E(T_r)}{w_r(e) \cdot x_e} \\
& \text{subject to}
& & \sum_{(s,t) \in D} \sum_{(\hat{s},\hat{t}) \in {R}_{s,t}^r} y_{\hat{s},\hat{t}} = 1 \label{eq:lpjunc_flow_1}\\ 
& & & \sum_{\hat{t}:(\hat{s},\hat{t}) \in R_{s,t}^r} y_{\hat{s},\hat{t}} \leq z_{\hat{s}} & \forall (s,t) \in D, \hat{s} \in S_{s,t}^r \label{eq:lpjunc_flow_2}\\ 
& & & \sum_{\hat{s}:(\hat{s},\hat{t}) \in R_{s,t}^r} y_{\hat{s},\hat{t}} \leq z_{\hat{t}} & \forall (s,t) \in D, \hat{t} \in T_{s,t}^r  \label{eq:lpjunc_flow_3}\\
& & & \text{ capacities } x_e \text{ support } z_{\hat{s}} \text{ flow from } \hat{s} \text{ to } \hat{r} & \forall (s,t) \in D, \hat{s} \in S_{s,t}^r \label{eq:lpjunc_flow_4}\\
& & & \text{ capacities } x_e \text{ support } z_{\hat{t}} \text{ flow from } \hat{r} \text{ to } \hat{t} & \forall (s,t) \in D, \hat{t} \in T_{s,t}^r  \label{eq:lpjunc_flow_5}\\
& & &  y_{\hat{s},\hat{t}} \geq 0 & \forall (s,t) \in D, (\hat{s},\hat{t}) \in R_{s,t}^r \\
& & & x_e \geq 0 & \forall e \in E(T_r)
\end{align}
\end{subequations}

\paragraph{Pruning the set of representatives:} 
We now need to prune the sets of representatives $S_{s,t}^r$ and $T_{s,t}^r$. This is to ensure that we can transform a solution to the above LP to an LP relaxation for GROUP STEINER TREE on the tree $T_r$ (which has an effective rounding scheme) where we can pick representatives independently while still fulfilling the resource constraints. Our overall approach here is similar to \cite{chlamtavc2020approximating} but we have to repeat their procedure for $m$ times taking a smaller and smaller set of potential terminals in each round. This means we would need a bigger scaling parameter, but if $m$ is small enough, it wouldn't affect the approximation factor significantly.

In the following discussion, when we use the term LP value, we mean $y_{\hat{s},\hat{t}}$ (where $(\hat{s},
\hat{t})\in R_{s,t}^r$). More formally, we need to find representative sets $\tilde{S}_{s,t}^r \subseteq S_{s,t}^r$ and $\tilde{T}_{s,t}^r \subseteq T_{s,t}^r$ such that $\tilde{S}_{s,t}^r \times \tilde{T}_{s,t}^r \subseteq R_{s,t}^r$ - this ensures that no matter what representative we pick on each set if they are connected - it will be within the resource requirements. The challenge here is to ensure we don't lose too much of the LP value i.e., $y_{\hat{s},\hat{t}}$ while we do so. 

\cite{chlamtavc2020approximating} accomplishes this by sorting the representatives of each terminal by their distance labels, and taking all representatives upto the median representative (by summation of the LP value). If we do implement this sorting procedure on both sides of $T_r$, then the representatives we pick will be such that any combination of them will satisfy the single distance requirement they have, while also having sufficient LP value. While the above pruning procedure would be sufficient for one distance constraint (for a given $(s,t)$ pair) - it cannot satisfy multiple distance constraints for a single $(s,t)$ pair simultaneously. In fact, it is entirely unclear what an analogous sorting procedure might be - for sorting based on a specific resource $ i \in \{1,\ldots,m\}$ does not in any way guarantee good performance for another resource $j \neq i$. Further, it is extremely important to ensure that when we take a set of good representatives for a specific resource $i$, we don't disallow the good representatives for another resource $j$. 

We overcome this challenge with a simple yet delicate adaption of \cite{chlamtavc2020approximating}'s pruning procedure which is described below. We first sort the terminals according to their resource consumption for the first resource alone (breaking ties arbitrarily). We then prune the representatives which consume more of the first resource than the consumption of the median terminal, and then do the same for the second resource - but only among the representatives still alive - and repeat this procedure for the other resource. Note the subtle but extremely important change here - we only prune after the median resource consumption, not the median terminal. If we prune after the median terminal, we might get rid of terminals that are actually good enough for this specific resource(since they might consume the same value of this resource as the median terminal). Also, note that repeated pruning does mean that the LP value of our final set of representatives is much smaller, but we can compensate by using a big enough scaling parameter (when $m$ is a constant, the scaling parameter is also a constant). 

This technique only works because our initial graph construction ensures that the good representatives for any resource $i$ (roughly speaking, a representative $\hat{s} \in S_{s,t}^r$ is good with respect to a resource $i$ if the $i^{th}$ resource consumption associated with $\hat{s}$ can potentially meet the budget requirements for $(s,t) \in D$) are evenly distributed across the length measures for any other resource $j$. That is, all representatives in $\tilde{S}_{s,t}^r$ which are good with respect to a resource $i$ are evenly distributed over every possible value of any different resource $j$. This ensures that the proportion of representatives that are good for resource $j$ will remain the same before and after running the pruning procedure with respect to resource $i$.

We describe the formal pruning procedure in Algorithm \ref{alg:junc_pruning}, which is designed for a specific $(s,t)$ pair. We can then run this procedure for every $(s,t)\in D$ one after the other. Algorithms \ref{alg:junc_sorting}, \ref{alg:junc_median} are minor subroutines used by Algorithm \ref{alg:junc_pruning}.

\begin{algorithm}[!htb]
\caption{Representative set sorter $(\tilde{S}_{s,t}^r,c)$} \label{alg:junc_sorting}
\begin{algorithmic}[1]

\State $\text{Sort representative sets } \tilde{S}_{s,t}^r \text{ and } \tilde{T}_{s,t}^r$ in the non decreasing order using the $c^{th}$ element of their distance label vector.

\State In other words, Sort 
\[\tilde{S}_{s,t}^r = \{(s_r^t,I_1),(s_r^t,I_2),\ldots\} \]
so that for any 
\[I_1,I_2 \text{ if } ( I_1=(z_1^1,z_2^1,\ldots,z_c^1,\ldots)) \text{ and } I_2 = (z_1^2,z_2^2,\ldots,z_c^2,\ldots))\]
if $z_c^1 < z_c^2$ then $(s_r^t,I_1)$ comes before $(s_r^t,I_2)$ in the sorted order. Ties are broken arbitrarily.

\State \Return $\tilde{S}_{s,t}^r$
\end{algorithmic}
\end{algorithm}

\begin{algorithm}[!htb]
\caption{Representative set median consumption finder $(S_{s,t}^r, T_{s,t}^r, \lambda, R_{s,t}^r,\text{ a solution to LP \eqref{lp:junc_flow} },c)$} \label{alg:junc_median}
\begin{algorithmic}[1]

\State{$\mu(s^t,c)$ is the median resource consumption that we are trying to find now. It only cares about the consumption of resource $c$.}

\State{$\mu(s^t,c) := \min\{ q \mid \sum\limits_{i=1}^{i=q} \sum\limits_{\hat{t}:((s_r^t,J),\hat{t}) \in R_{s,t}^r, J[c]=i} y_{((s_r^t,J),\hat{t})}
\geq \lambda\}$} 

\State \Comment{The above line just tries to find the first value of resource consumption $\mu(s^t,c)$ for resource $c$ such that such that the sum of LP values for all terminals that use less than or equal to $\mu(s^t,c)$ of resource $c$ in the sorting is at least $\lambda$.}

\State{$\mu(t^s,c) := \min\{ q \mid \sum\limits_{i=1}^{i=q} \sum\limits_{\hat{s}:(\hat{s},(t_r^s,J)) \in R_{s,t}^r, J[c]=i} y_{(\hat{s},(t_r^s,J))} \geq \lambda\}$}

\State \Return $\mu(s^t,c),\mu(t^s,c)$
\end{algorithmic}
\end{algorithm}

\begin{algorithm}[!htb]
\caption{Representative set pruner $(S_{s,t}^r, T_{s,t}^r, m, \text{ a solution to LP \eqref{lp:junc_flow} },R_{s,t}^r)$} \label{alg:junc_pruning}
\begin{algorithmic}[1]

\State{$\tilde{S}_{s,t}^r \gets S_{s,t}^r$, $\tilde{T}_{s,t}^r \gets T_{s,t}^r$.}

\State{$\gamma_{s,t} \gets \sum_{(\hat{s},
\hat{t})\in R_{s,t}^r}y_{\hat{s},\hat{t}}$.}

\For{$c \in 0,1,2,\ldots,m$} \Comment{Go through all resource one at a time}

\State $\tilde{S}_{s,t}^r = \text{ Algorithm \ref{alg:junc_sorting} }(\tilde{S}_{s,t}^r,c)$ \Comment{Sort using Algorithm \ref{alg:junc_sorting}}.
 
\State $\tilde{T}_{s,t}^r = \text{ Algorithm \ref{alg:junc_sorting} }(\tilde{T}_{s,t}^r,c)$ .

\State $\mu(s^t,c),\mu(t^s,c) \gets \text{ Algorithm \ref{alg:junc_median} }(\tilde{S}_{s,t}^r,\tilde{T}_{s,t}^r,\gamma_{s,t}/2^c,R_{s,t}^r,\text{ a solution to LP \ref{lp:junc_flow} },c)$. \Comment{Find median resource consumption using \ref{alg:junc_median}}

\State{$\tilde{S}_{s,t}^r \gets \text{ All terminals in } \tilde{S}_{s,t}^r \text{ that use less than or equal to } \mu(s^t,c)$ \text{ of resource } c}

\State{$\tilde{T}_{s,t}^r \gets \text{ All terminals in } \tilde{T}_{s,t}^r \text{ that use less than or equal to } \mu(t^s,c)$ \text{ of resource } c}

\EndFor

\State \Return $\tilde{S}_{s,t}^r,\tilde{T}_{s,t}^r$
\end{algorithmic}
\end{algorithm}
The following lemma shows that Algorithm \ref{alg:junc_pruning} does not allow any terminal pairs representing $(s,t) \in D$ that cannot fulfill the resource budget limit for $(s,t)$ to stay alive after all iterations.

\begin{lemma} \label{le:junc_pruner}
    After running Algorithm \ref{alg:junc_pruning}, we have $\tilde{S}_{s,t}^r \times \tilde{T}_{s,t}^r \subseteq R_{s,t}^r$.
\end{lemma}

\begin{proof}
    We will show that after the $c^{th}$ iteration, all terminal pairs that stay alive in $\tilde{S}_{s,t}^r$ satisfy the resource budget constraint for resource $c$. That way, any terminal pair that is alive after all the iterations will satisfy the budget requirements for all the resource. 

    Recall that $\mu(s^t,c)$ is the median resource consumption (among sources) for resource $c$ and that we allow any terminal (among sources) consuming resource $\leq \mu(s^t,c)$ for type $c$ to stay alive. For the $c^{th}$ iteration, let $U_c$ denote the set of terminals (among $\tilde{S}_{s,t}^r$) that consume $\geq \mu(s^t,c)$ units of resource $c$. From our choice of $\mu(s^t,c)$ in Algorithm \ref{alg:junc_pruning}, we can see that the total sum of the LP value for the source terminals in $U_c$ is at least $\gamma_{s,t}/2^c$.

    Let $c_{max}$ be the maximum value of resource $c$ that can correspond to a {\em valid} (meaning can be connected within resource budget) sink terminal for any source terminal which consumes $\mu(s^t,c)$ units of resource $c$. All terminals in $U_c$ have resource consumption for resource $c$ at least $\mu(s^t,c)$. If these terminals happen to be a part of some {\em valid} terminal pair, then it can only be with some sink terminal that uses resources less than or equal to $c_{max}$. 

    Thus, summing up the LP value for all the sinks that consume resources less than or equal to $c_{max}$, we should have total LP value at least $\gamma_{s,t}/2^c$ (because any terminal in $U_c$ can send flow only to these terminals, and terminals in $U_c$ have flow at least $\gamma_{s,t}/2^c$). This ensures that $c_{max} \geq \mu(t^s)$ by definition of $\mu(t^s,c)$ from Algorithm \ref{alg:junc_pruning}. Thus, $\mu(s^t,c) + \mu(t^s,c) \leq \mu(s^t,c) + c_{max} \leq \bdgt(s,t)[c]$.

    All terminals still alive in $\tilde{S}_{s,t}^r$ consume $\leq \mu(s^t,c)$ units of resource $c$. Thus, they can all be connected to any sink terminal that can be connected to a source consuming $\mu(s^t,c)$ units of resource $c$. Similarly, we can see that all terminals that are still alive in $\tilde{T}_{s,t}^r$ consume $\leq \mu(t^s,c)$ units of resource $c$. Thus, they can all be connected to any source terminal that can be connected to a sink consuming $\mu(t^s,c)$ units of resource $c$. Since $\mu(s^t,c) + \mu(t^s,c) \leq \bdgt(s,t)[c]$, this means any terminal pair that is alive after iteration $c$, can be connected within the budget requirements for the resource $c$.
\end{proof}

\paragraph{Running \gsr:}
We will now try to run the \gsr algorithm on this pruned set of representatives and show that achieves our objective. By the definition of $\gamma_{s,t}$ from Algorithm \ref{alg:junc_pruning} and the constraint \eqref{eq:lpjunc_flow_1} in LP \eqref{lp:junc_flow}, we have $\sum_{(s,t)\in D} \gamma_{s,t} = 1$. We will now bucket the pairs in the demand pair set $D$ by their $\gamma_{s,t}$ values with $\lceil log|D| \rceil$  buckets. 

\begin{align}
    D_i = \{(s,t) \in D \mid \gamma_{s,t} \in [(]2^{-i-1},2^{-i}]\}.
\end{align}

By a standard counting argument, there is some bucket $i^*$ such that $\sum_{(s,t)\in D_{i^*}}\gamma_{(s,t)} \geq \frac{1}{2(\lceil log|D| \rceil +1)}$, and so $|D_{i^*}| \geq 2^{i^*}/O(\log n)$. In addition, for every pair in this bucket, i.e., for every $(s,t) \in D_{i^*}$, we have (using constraint \eqref{eq:lpjunc_flow_2}),

\begin{align}
    \sum_{\hat{s} \in \tilde{S}_{s,t}^r}z_{\hat{s}} \geq \sum_{\hat{s} \in \tilde{S}_{s,t}^r}\sum_{\hat{t}:(\hat{s},\hat{t}) \in R_{s,t}^r} y_{\hat{s},\hat{t}} \\
    \geq \gamma_{s,t}/ 2^{m+1} \geq 2^{-i^* - 1}/2^{m+1}
\end{align}

and in the same way,
\begin{align}
    \sum_{\hat{t} \in \tilde{T}_{s,t}^r}z_{\hat{t}}  \geq \gamma_{s,t}/2^{m+1} \geq 2^{-i^* - 1}/2^{m+1}
\end{align}

Now, we will scale our solution by $x_e^* := \min \{1,2^{m+1} \cdot 2^{i^* + 1}\cdot x_e\}$, and from the constraints \eqref{eq:lpjunc_flow_4} and \eqref{eq:lpjunc_flow_5} along with the above bound, we get a (potentially suboptimal) solution to the following LP.

\begin{equation}
\begin{aligned} \label{lp:junc_gsr}
& \min & & \sum_{e \in E    }{w(e) \cdot x_e^*} \\
& \text{subject to}
& & \text{ capacities } \{x_e^*\} \text{ support one unit of flow from } \tilde{S}_{s,t}^r \text{ to } \hat{r} & \forall (s,t) \in D_{i^*}\\
& & & \text{ capacities } \{x_e^*\} \text{ support one unit of flow from } \hat{r} \text{ to } \tilde{T}_{s,t}^r   & \forall (s,t) \in D_{i^*}\\
& & & x_e^* \geq 0 & \forall e \in E
\end{aligned}
\end{equation}

The following lemma is from (\cite{chlamtavc2020approximating,chekuri2011set})
\begin{restatable}{lemma}{lemmagstr}
    \label{le:group_Steiner_rounding} \cite{chlamtavc2020approximating}
    (\textbf{Group Steiner Tree rounding in a tree}) Given an edge-weighted undirected tree $T$ with $n$ vertices rooted at $r$, with edge weights $w:E(T)\rightarrow R^{\geq 0}$, a collection of vertex sets $\hat{S} \in \text{ Power set }(V(T))$, and a solution to the following LP:
\begin{equation}
\begin{aligned} \label{lp:group_rounding}
& \min & & \sum_{e \in E}{w(e) \cdot x_e} \\
& \text{subject to}
& & \text{ capacities } x_e \text{ support one unit of flow from } r \text{ to } S & \forall S \in \hat{S}\\
& & & x_e \geq 0 & \forall e \in E
\
\end{aligned}
\end{equation}
Then we can efficiently find a subtree $T' \subseteq T$ rooted at $r$ such that for every $S \in S$ at least one
vertex of $S$ participates in the tree $T'$, and $w(T')\leq O(\log n \log |\hat{S}| · \sum_{e \in E}w(e) x_e)$.  
\end{restatable}

By Lemma \ref{le:group_Steiner_rounding}, we can round the solution to LP \ref{lp:junc_gsr}, and obtain a tree $T'\subseteq T_r$ of weight 
\begin{align}
    w(T') &= O(\log^2 \left(n \cdot \Pi_{i=1}^{i=m+1}\left(\boldsymbol{t^{-}}[i] + \boldsymbol{t^{+}}[i] \right)\right) \cdot \sum_{e \in E(T_r)}{w(e) \cdot x_e^*})
    \notag
    \\
    &\leq 2^{m+1} \cdot 2^{i^*} O(\log^2 \left(n \cdot \Pi_{i=1}^{i=m+1}\left(\boldsymbol{t^{-}}[i] + \boldsymbol{t^{+}}[i] \right)\right) \cdot \sum_{e \in E(T_r)}{w(e) \cdot x_e})
    \notag
\end{align}

such that for every pair $(s,t) \in D_{i^*}$ at least one vertex $\hat{s} \in \tilde{S}_{s,t}^r$ and at least one vertex $\hat{t} \in \tilde{T}_{s,t}^r$ are connected through $r$ in the tree $T'$. Since these pairs also satisfy resource constraints (As per Lemma \ref{le:junc_pruner}), we can use the following equation for a bound on the density of the tree $T'$:

\begin{align}
    \frac{w(T')}{\mid\{(s,t)\in D \mid \exists (\hat{s},\hat{t}) \in 
 R_{s,t}^r:\text{ $T'$ contains an $\hat{s} \leadsto \hat{t}$ walk }\}|} \leq \frac{w(T')}{|D_{i^*}|} 
 \\= 2^{m+1} \cdot O(\log^3 \left(n \cdot \Pi_{i=1}^{i=m+1}\left(\boldsymbol{t^{-}}[i] + \boldsymbol{t^{+}}[i] \right)\right) \cdot \sum_{e \in E(T_r)}{w(e) \cdot x_e})
 \notag
\end{align}

which rounds our original LP, and proves Lemma \ref{le:junction_tree_intermediate}, and in turn our overall result. Note that the runtime will still remain polynomial in $n \cdot \Pi_{i=0}^{i=m}(|\boldsymbol{t^{-}}[i]| + |\boldsymbol{t^{+}}[i]|  )$.

\section{An approximation algorithm for \pks (Integer version)}
\label{sec:khalf_vanilla}

Recall that $k=|D|$ is the number of demand pairs for the given insance of \pks. We now see a proof of Theorem \ref{thm:vpks_khalf}.

\thmvpkskhalf*

\begin{proof}
    We first introduce the notion of resource-constrained junction tree solutions.

    \begin{definition}
A \emph{resource-constrained junction tree solution} is a collection of resource-constrained junction trees rooted at different vertices, such that the distance constraints for all $(s,t) \in D$ are satisfied.
\end{definition}

Our proof has two main ingredients. First, 
we construct a resource-constrained junction tree solution and compare its objective with the optimal resource-constrained junction tree solution with objective value $\opt_{junc}$. Theorem \ref{thm:mdrcjt_vanilla} implies a $\poly(n,1/\ep,\tau)$-time algorithm that finds a resource-constrained junction tree solution of cost at most $\tO(k^\ep \tau^m) \opt_{junc}$. Second, we show the existence of an $O(\sqrt{k})$-approximate solution consisting of resource-constrained junction trees, i.e., $\opt_{junc} \le O(\sqrt{k}) \opt$ where $\opt$ is the optimal cost. Combining these two ingredients implies Theorem \ref{thm:vpks_khalf}.

We use a density argument via a greedy procedure which implies an $O(\sqrt{k})$-approximate resource-constrained junction tree solution. We recall that the density of a resource-constrained junction tree is its cost divided by the number of terminal pairs that it connects by feasible walks.

Intuitively, we are interested in finding low-density resource-constrained junction trees. We show that there exists a resource-constrained junction tree with density at most an $O(\sqrt{k})$ factor of the optimal density. The proof of Lemma~\ref{lem:sqrt-k-den} closely follows the one for the directed Steiner network problem in \cite{chekuri2011set}, pairwise spanners \cite{grigorescu2021online}, and weighted spanners \cite{gkl2023}, by considering whether there is a \emph{heavy} vertex that lies on $s \leadsto t$ walks for $(s,t) \in D$ or there is a walk with low density. The case analysis also holds with resource constraints.

\begin{restatable}{lemma}{lemsqrtkden} \label{lem:sqrt-k-den}
There exists a resource-constrained junction tree $\cJ$ with density at most $\opt / \sqrt{k}$.
\end{restatable}

\begin{proof}
Let $H$ be the optimal solution subgraph with cost $\opt$. The proof proceeds by considering the following two cases: 1) there exists a vertex $r \in V$ that belongs to at least $\sqrt{k}$ $s \leadsto t$ feasible walks for distinct $(s,t)$ using edges in $H$, and 2) there is no such vertex $r \in V$.

For the first case, we consider the union of the $s \leadsto t$ walks using edges in $H$ that pass through $r$. This forms a subgraph spanned by these $s \leadsto r \leadsto t$ walks, whose union forms a resource-constrained junction tree rooted at $r$. This resource-constrained junction tree has cost at most $\opt$ and connects at least $\sqrt{k}$ terminal pairs, so its density is at most $\opt / \sqrt{k}$.

For the second case, each vertex $r \in V$ appears in at most $\sqrt{k}$ $s \leadsto t$ walks in $H$. More specifically, each edge $e \in E$ also appears in at most $\sqrt{k}$ $s \leadsto t$ walks in $H$. By creating $\sqrt{k}$ copies of each edge with the same cost $\sigma$, all terminal pairs can be connected by edge-disjoint walks (for each $s \leadsto t$ walk, if edge $e$ is used multiple times, then the cost $\sigma(e)$ is counted only once). Since the overall duplicate cost is at most $\sqrt{k} \cdot \opt$, at least one of these walks has cost at most $\sqrt{k} \cdot \opt / k $. This walk constitutes a resource-constrained junction tree whose density is at most $\opt / \sqrt{k}$.
\end{proof}

Consider an iterative procedure that finds a minimum density resource-constrained junction tree and continues on the remaining disconnected terminal pairs. Suppose there are $t$ iterations, and after iteration $j \in [t]$, there are $n_j$ disconnected terminal pairs. For notation convenience, let $n_0 = k$ and $n_t = 0$. After each iteration, the minimum cost for connecting the remaining terminal pairs in the remaining graph is at most $\opt$, so the total cost of this procedure is upper-bounded by
\[
\sum_{j=1}^t \frac{(n_{j-1} - n_j)\opt}{\sqrt{n_{j-1}}} \leq \sum_{i=1}^k \frac{\opt}{\sqrt{i}} \leq \int_1^{k+1} \frac{\opt}{\sqrt{x}} dx = 2 \opt (\sqrt{k+1} - 1) = O(\sqrt{k}) \opt\]
where the first inequality uses the upper bound by considering the worst case when only one terminal pair is removed in each iteration of the procedure.
\end{proof}
\section{Applications of \pks}
\subsection{\rcs}
\label{sec:rcs}

Recall Definition \ref{def:rcs}.
\defrcs*

Before we prove Theorem \ref{thm:rcs_khalf}, we establish the following helper claim.

\begin{claim} \label{cl:walk_visit_count}
    For any {\em routing-controlled feasible} walk $p(s,t)$ that resolves a demand pair $(s,t) \in D$, there exists a {\em routing-controlled feasible} walk $p'(s,t)$ (not necessarily distinct from $p(s,t)$) that visits any point in the graph at most $c$ times. Further the cost of $p'(s,t)$ is at most the cost of $p(s,t)$
\end{claim}
\begin{proof}
    First, observe that since all lengths are positive in \rcs, the only reason to revisit a vertex is that we may have to go to multiple must-visit vertices. Next, any walk resolving $(s,t)$ only needs to visit at most $c$ must-visit vertices (due to the fact that we have at most $c$ sets we need to visit and we need to visit only one point in each set). For reaching a specific must-visit vertex, we don't need to visit any intermediate point more than once. If we do revisit the same intermediate point twice to visit the same must-visit vertex, then we have cycle which does not reach any new must-visit vertex and that cycle can be deleted. Thus, we don't need to visit any point in the graph more than $c$ times.
\end{proof}

Recall Theorem \ref{thm:rcs_khalf}.

\thmrcskhalf*

\begin{proof}
    The key idea here is to reduce a given instance of \rcs into a instance of \pks. Overall we will trying to use one resource for each group in \rcs; that resource is responsible for ensuring we visit/don't visit that specific set. Then it is just a matter of using Theorem \ref{thm:vpks_khalf} to obtain this result.

    Let $I_1$ be the given instance of \rcs with the following parameters:
    \begin{itemize}
        \item Graph $G_1=(V,E)$.
        \item Edge costs $\sigma(e)$.
        \item Edge length $\ell_e$.
        \item Groups - $S_1,S_2,\ldots,S_p,S_{p+1},\ldots,S_{p+c}$
        \item Demand pair set $D_1$ and a control function $\ctrl$.
    \end{itemize}

    $I_1$ can be reduced to the following instance $I_2$ of \pks with the parameters:
    \begin{itemize}
        \item A graph $G_2=(V,E)$ (this is a new copy with the same set of vertices and edges). The same value of $m$ as $I_1$.
        \item The threshold $\threshold = m \cdot max_{S_{p+1} \ldots S_{p+c}}|S_i|$. Note that because $m$ and $max_{S_{p+1} \ldots S_{p+c}}|S_i|$ are constants, $\tau$ is also a constant.
        \item The same edge costs $\sigma(e)$.
        \item Edge consumption : $\boldsymbol{r_e}$ with $\boldsymbol{r_e}[0] = \ell_e$. 
        \begin{itemize}
            \item When $S_i$ is a set of vertices that need to be visited by some demand pairs: Let $e=(u,v)$, then when $v \in S_{i}$, $\boldsymbol{r_e}[i] = -1$. When $v \not\in S_{i}$, $\boldsymbol{r_e}[i] = 0$. 
            \item When $S_i$ is a set of vertices that need to be avoided by some demand pairs:  Let $e=(u,v)$, then when $v \in S_{i}$, $\boldsymbol{r_e}[i] = 1$.  when $e \not\in S_i$, $\boldsymbol{r_e}[i] = 0$. 
        \end{itemize} 
        \item Demand pairs: $D_2 = \{(s,t) \:\:\: \forall \: (s,t) \in D_1\}$.
        \item Budgets: 
        \begin{itemize}
            \item We set $\bdgt(s,t)[0] = \ctrl(s,t)[0]$.
            \item When $i \in [1,c]$, we set $\bdgt(s,t)[i] = -1$ when $\ctrl(s,t)[i] = 1$. We set $\bdgt(s,t)[i] = 0$ when $\ctrl(s,t)[i] = 0$.

            The idea here is that, when $\ctrl(s,t)[i] = 1$, we are forced to visit a vertex in $S_i$. By setting $\bdgt(s,t)[i] = -1$, a walk can only become feasible if it has at least one edge which goes through a vertex in $S_i$. If $\ctrl(s,t)[i] = 0$, any walk is feasible with respect to the resource $i$.
            
            \item When $i \in [c + 1, m]$, We set $\bdgt(s,t)[i] = m \cdot |S_i|$ when $\ctrl(s,t)[i] = 0$. We set $\bdgt(s,t)[i] = 0$ when $\ctrl(s,t)[i] = -1$. 
            
            The idea here is that when $\ctrl(s,t)[i] = -1$, we are forbidden from visiting any vertex in $S_i$ and thus $\bdgt(s,t)[i]$ is set to zero which prevents us from going to such a vertex in \pks (since visiting those vertices will consume the resource $i$). When $\ctrl(s,t)[i] = 0$, we are allowed to visit the set $S_i$. Because of Claim \ref{cl:walk_visit_count}, if there is a {\em routing-controlled feasible} walk for $(s,t)$, then there is a {\em routing-controlled feasible} walk which visits any vertex in $S_i$ atmost $c$ times. Such a walk would use at most $c$ units of resource $i$ for any specific vertex in $S_i$. Across all such vertices, it would use at most $c \cdot |S_i|$ units of resource $i$. Every time we visit a vertex in $S_i$, the resource consumption increases by $1$ for the $i^{th}$ resource. This can only happen $c \cdot |S_i|$ times and thus a  budget of $c \cdot |S_i|$ is sufficient to have a feasible walk in \pks if one existed originally in \rcs.
        \end{itemize}

    \end{itemize}

    \begin{claim}
      For any $f \geq 0$, there exists a sub graph $H_1$ of  $G_1$  with $\cost(H_1) = f$ containing a {\em group-steiner-feasible} walk from $s$ to $t$ for every $(s,t) \in D_1$ if and only if there exists a sub graph $H_2$ of $G_2$ with $\cost(G_2) = f$ containing a walk of resource consumption $\bdgt(s,t)$ from $s$ to $t$ for every $(s,t) \in D_2$.  Moreover given a sub graph $H_2$ which is solves $I_2$ one can recover a sub graph $H_1$ which solves $I_1$ in poly-time.
       
    \end{claim}

    Now, we can just apply Theorem \ref{thm:vpks_khalf} to solve $I_2$ and then turn the solution of $I_2$ into a solution for $I_1$.
\end{proof}

For a given instance of \rcs, let  $\indexlist \subseteq \{1,2,\ldots,m\}$.
We say $S' = \cup_{i \in \indexlist}S_i$ is essential if every demand pair $(s,t) \in D$ has to visit at least one vertex from $S'$. Theorem \ref{thm:rcs_keps} presents a specialized result of \rcs in terms of the size of the set $S'$. When $|S'|$ is small, the Theorem gives a very good approximation factor (especially when the essential subset is constant size). Recall Theorem \ref{thm:rcs_keps}.

\thmrcskeps*

\begin{proof}
    This result relies on the fact that  the optimal solution - $\opt$ would also be forced to go through $S'$. Therefore the optimal solution can be considered as a combination of $|S'|$ distinct junction trees (each of cost $\leq \opt$) and we can approximate all of them separately using Theorem \ref{thm:mdrcjt_vanilla}.

    The rest of the proof of this theorem is very similar to that of Theorem \ref{thm:rcs_khalf} and therefore omitted.
\end{proof}

\subsection{As a tool for designing hopsets}
\label{sec:gbh}

A \emph{hopset} $H$ with hopbound $\beta$ and stretch $\alpha$ for a graph $G$ is a set of edges such that, 
for every pair of vertices $u, v \in V(G)$, there exists a path in $G \cup H$ that uses at most $\beta$ hops 
and whose length is at most $\alpha$ times the actual distance between $u$ and $v$ in $G$. \cite{dinitz2025approximation} studies the optimization version of this problem, where the goal for a given graph instance is to determine the smallest set of edges that meets the hopbound $\beta$ and stretch $\alpha$ requirements.

\begin{restatable}{definition}{defgenhopset}
\label{def:gen_hopset}

    \gbh    

    \textbf{Instance}: 
    \begin{enumerate}
        \item \textbf{Graph:} A directed graph $G = (V, E)$.
        
        \item \textbf{Edge lengths:} A length function $\ell : E \to [\poly(n)]_+$. 
        
        \item \textbf{Demand pairs:} A set $D \subseteq V \times V$ of ordered vertex pairs.
        
        \item \textbf{Distance bounds:} A function $\Dist : D \to \Z_{\geq 0}$ specifying, for each $(s,t) \in D$, an upper bound on the allowed distance.
        
       \item \textbf{Hopbound:} A parameter $\beta \in \Z_{>0}$ limiting the number of hops (edges) allowed on any path.
    \end{enumerate}

    \textbf{Objective}: Find a minimum-size edge set $H$ such that for every demand pair $(s,t) \in D$, there exists an $s \leadsto t$ path $P$ in $G \cup H$ with at most $\beta$ hops and total length 
    \[
        \sum_{e \in P} \ell(e) \leq \Dist(s,t).
    \]
\end{restatable}

We now show the versatility and the utility of the \pks formulation by recovering an Theorem 5.1 from \cite{dinitz2025approximation} using Theorem \ref{thm:vpks_khalf}.  We first state Theorem 5.1 from \cite{dinitz2025approximation}.

\begin{theorem}
There is a polynomial-time $\tilde{O}(\beta n^{\epsilon} \cdot \mathrm{OPT})$-approximation 
for the directed Generalized $\beta$-Hopset problem.
\end{theorem}

\textbf{Weighted Transitive Closure.} The weighted transitive closure of a graph $G$, denoted $G^M = (V, E^M)$, is created in two steps. First, we compute the transitive closure of $G$, connecting all pairs of vertices that are reachable from one another. Then, we assign to each edge $(u, v)$ in this closure a weight equal to the shortest distance $d_G(u, v)$ in the original graph.

One key idea in \cite{dinitz2025approximation} is to first find the weighted transitive closure. Then, every edge is assigned a cost on top of the weights; this cost is $1$ if the edge is not in the original graph and $0$ otherwise. After this, we just need to find the min cost subgraph in the transitive closure that resolves all demand pairs. 

Observe that finding this min cost subgraph is just a subproblem of \pks where we have only one resource other than length i.e., m = 1. The magnitude threshold is $\beta$ and the resource consumption for the first resource is always 1 and the resource consumption for the zeroth resource is given by the length of the edge. The resource budget is a 2 element vector where the first element is the distance constraint and the second element is always $\beta$. Observe that $\opt \geq k^{1/2}$ because with $\opt$ edges we can have atmost $\opt$ distinct sources and sinks (see Lemma 17 in \cite{grigorescu2021online} for more details). Then we can recover Theorem 5.1 from \cite{dinitz2025approximation} immediately by using Theorem \ref{thm:vpks_khalf}.

\thmvpkskhalf*

In fact, we can use Theorem \ref{thm:khalf_pks} to extend the result to accommodate rational and even negative lengths. One could even add the routing controls on top of the hopbound and Theorem \ref{thm:khalf_pks} can handle it. One could even set a different hop bound for each terminal pair where the $\beta = max_{(s_i,t_i)} \beta_i$. Note that to generalize the overall result in \cite{dinitz2025approximation} would require us to generalize other tools in \cite{dinitz2025approximation} as well. While we do believe that it may be possible in at least some of the tools used by \cite{dinitz2025approximation}, this overall process is beyond the scope of our work. This section just illustrates the myriad ways in which \pks can be useful.

\section{Proof of \rcjt (Rational version)}

\label{sec:jtree_complex_version} 

Recall Definition \ref{def:pks}.

\defpks*

We now consider \pks by slightly relaxing the distance constraints. For notation convenience, we define some condition numbers. Let 

\begin{equation} \label{def:eta}
    \eta := \frac{|\min\{\min_{e \in E}\{\boldsymbol{r_e}[0]\},0\}|}{\min_{(s,t) \in D} \{|\bdgt[0]|\}}.
\end{equation}

and let $\xi$ is assigned as follows:

\begin{equation} \label{def:xi}
    \xi:= \frac{\max_{(s,t) \in D} \{|\bdgt[0]|\}}{\min_{(s,t) \in D} \{|\bdgt[0]|\}}.
\end{equation}

These condition numbers are only required for the zeroth resource. Intuitively, $\eta$ denotes the ratio between the magnitude of the most negative resource consumption (for the zeroth resource) and the smallest absolute value of the budget (for the zeroth resource).\footnote{$\eta$ cares about $\min_{e \in E}\{\ell_e\}$, but not about $\max_{e \in E}\{\ell_e\}$. This is because we can safely ignore edges that have consumption bigger than the budget, but we cannot do so for edges (with negative lengths) that consume much lesser than the budget.} If the resource consumption (for the zeroth resource) is always non-negative, then $\eta=0$. Similarly, $\xi$ denotes the ratio between the largest and the smallest absolute value of the budget (for the zeroth resource).

To accommodate a negative distance budget, we use the $\siign$ function $\siign(x) = \begin{cases} 
                 -1  & \text{if } x < 0, \\
                 0 & \text{if } x = 0, \\
                 1 & \text{if } x > 0.
            \end{cases}$

Suppose we are given a tolerance parameter $\theta > 0$. When the distance budget is positive, the goal is to ensure that the consumption of the zeroth resource satisfy the its constraint within a factor of $(1+\theta)$. When the distance budget is negative, the goal is to ensure that the consumption of the zeroth resource is below $(1-\theta)$ times the budget. The last $m$ resources always need to be exactly within the given budget. This is formally defined in 
 Definition \ref{def:theta_feasible_path}.

\begin{restatable}[$\theta$-Feasibility]{definition}{defthetafeasiblepath}\label{def:theta_feasible_path}
For a given demand pair $(s,t) \in D$, we say that a $s \leadsto t$ walk $p(s,t)$ is {\em $\theta$-feasible} if 
\begin{itemize}
    \item $\sum_{ e \in p(s,t)} \boldsymbol{r_{e}}[0] \le \bdgt(s,t)[0] \cdot (1+\theta \cdot \siign(\bdgt(s,t)[0])) $ and
    \item $\sum_{ e \in p(s,t)} \boldsymbol{r_{e}}[i] \le \bdgt(s,t)[i] \:\: \forall i \in [1,m]$.
\end{itemize}
\end{restatable}

We now state the formal definition of a relaxed version of junction trees.

\begin{restatable}{definition}{defthetajuncbbs} \label{def:junc_complex}
    \rrcjt 
    
    \textbf{Instance}: Same as \pks.
    
    \textbf{Objective}: Find a root $r \in V$, a resource-constrained junction tree $\cJ$ rooted at $r$, such that the ratio of the cost of $\cJ$ to the number of $(s,t)$ pairs that satisfy the $\theta$-feasibility requirement (recall Definition \ref{def:theta_feasible_path})
    is minimized. Specifically, the goal is the following:
    \begin{equation}
             \min_{r \in V, \cJ}\frac{\cost(\cJ)}{|\{(s,t) \in D \mid \exists \: \text{a $\theta$-feasible walk } p(s,r,t) \text{ in } \cJ| }.
    \end{equation}
\end{restatable}

Here is the rational version of our main theorem for finding the minimum density \rcjt.

\begin{restatable}{theorem}{thmmdrcjtcomplex} \label{thm:mdrcjt_complex}.
(Minimum density Resource constrained junction tree approximatin - Rational version) 

Let $k = |D|$ where $D$ is the set of demand pairs in the given instance of \pks. Then, for any $\theta > 0$,  constants $\ep > 0, m > 0$, there is a $\poly(n,\xi,\eta,1/\theta,\tau^m,1/\ep)$-time algorithm that gives a $\tO(k^\ep \cdot \polylog( ((\eta+\xi)/\theta) \cdot \threshold^m) \cdot \threshold^m)$\footnote{note that $\eta,\xi$ and $\tau$ can be exponential of input size and therefore cannot be represented as $\polylog(n)$}-approximation for \trcjt. When $1/\theta,\eta,\xi,\threshold, \hop \in \poly(n)$, the algorithm runs in polynomial time.

\end{restatable}

We note that when the edge lengths are non-negative, $\eta = 0$, so the running time is $\poly(n,1/\theta,\xi)$.

\subsection{Proof of Theorem \ref{thm:mdrcjt_complex}}
    
\subsubsection{High-level Idea}

We do the following procedure using every possible root $r \in V$ and then take whichever case has the minimum density among all the possible roots. First, we scale the resource consumptions of the zeroth resource and turn them into integers. This step comes at the cost of only being able to approximately satisfy the zeroth resource. Then we build a product graph to turn the scaled resource constraints/budgets into connectivity constraints. This gives rise to an instance of the \mslc problem. Then we use the height reduction result from \cite{zelikovsky1997series} to obtain a shallow tree  for the instance of \mslc. Finally, we can solve this instance by using the \gsf problem using an extension of the approach in \cite{chlamtavc2020approximating}.

We only detail the scaling procedure in this proof. The rest of this proof is identical to the proof of Theorem \ref{thm:mdrcjt_vanilla} and therefore omitted. 

\subsubsection{Scaling the weights}

    The overall scaling approach here is based on \cite{grigorescu2024directed,horvath2018multi}. Let the resource consumption of $P_{s,t}$ be denoted by $\res(P_{s,t})$. Let $\bdgt_{max} =  \max_{(s,t) \in D}| {\bdgt(s,t)[0]} |$ and $\bdgt_{min} = \min_{(s,t) \in D}  | {\bdgt(s,t)[0]} |$. Also, let $\lenlow = \min_{e \in E} \boldsymbol{r_e}[0] \text{ and } \lenup = \max_{e \in E} \boldsymbol{r_e}[0]$. In the given instance of \pks, let \hop be the smallest positive integer such that for every $(s,t) \in D$, there exists a feasible $s \leadsto t$ walk $p(s,t)$ with fewer than \hop edges. Recall that we assume that there are no negative cycles for the zeroth resource and the budget for the last $m$ resources is always a constant. Therefore, there exists a positive constant $c$ with $\hop \leq c \cdot \poly(n)$. Let $\Delta = \theta \cdot (\bdgt_{min}) /\hop$: intuitively $\Delta$ is a measure of the level of precision we need when we scale the edge lengths. 

    Here is the scaling procedure: for any edge $e \in E$, we set $\sres_{e}[1] = d_{e} \cdot \Delta^1$ where $d_{e}$ is some integer which ensures $(d_{e} - 1) \cdot \Delta < \boldsymbol{r}_e[0] \leq d_{e} \cdot \Delta$ is true. The last $m$ resources undergo no change due to the scaling procedure i.e., $\sres_{e}[i] = \boldsymbol{r_e}[i] \; \forall i \in [1,m], e \in E$. We call this new graph with the scaled edge resource consumption as the scaled graph $\bar{G}$. $\bar{G}$ has the vertex set $V$ and edge set $E$, but the resource consumptions in $\bar{G}$ are set to $\sres_e$. The costs for the edges in $\bar{G}$ are inherited from the corresponding edges in $G$.

    In the following discussion, we are going to slightly abuse notation and use the same variable to denote a walk in both the original and scaled graph. The idea is that the scaled version of a walk from the original graph is another walk with the corresponding sequence of vertices. The only thing that will change is the function used to calculate the resource consumption of the walk in the original and scaled graph.

\begin{figure}[H]
\begin{subfigure}[b]{.45\linewidth}
\centering
\begin{tikzpicture}[scale=0.5]
    \node[fill,circle, inner sep=0pt, minimum size=0.2cm] (a) at (-3,-3) {};
    \node[fill,circle, inner sep=0pt, minimum size=0.2cm] (c) at (-3,3) {};
    \node[fill,circle, inner sep=0pt, minimum size=0.2cm] (b) at (4,-3) {};
    \node[fill,circle, inner sep=0pt, minimum size=0.2cm] (d) at (4,3) {};
    \node[fill,circle, inner sep=0pt, minimum size=0.2cm] (r) at (0.5,6) {};
    \node at (-4,-3.7) {$a$};
    \node at (-4,3.7) {$c$};
    \node at (5,-3.7) {$b$};
    \node at (5,3.7) {$d$};
    \node at (0.5,6.7) {$r$};

    \path[->] (a) edge node[below,align = center] {\footnotesize {\color{black} $2$} }  (b);
  
   \path[->] (b) edge node[below,align = center, rotate=300, pos=0.4] {\footnotesize {\color{black} $3$} }  (r);
   \path[->] (c) edge node[left,align = center] {\footnotesize {\color{black} $1$} }  (a);
   \path[->] (c) edge node[below,align = center] {\footnotesize {\color{black} $2$} }  (d);
   \path[->] (d) edge node[right,align = center] {\footnotesize {\color{black} $1$} }  (b);
   \path[->] (d) edge node[above,align = center, rotate=315] {\footnotesize {\color{black} $2$} }  (r);
   \path[->] (r) edge node[below,align = center, rotate= 60, pos=0.6] {\footnotesize {\color{black} $4$} }  (a);
   \path[->] (r) edge node[above,align = center, rotate=45] {\footnotesize {\color{black} $-3$} }  (c);
\end{tikzpicture}
\subcaption{The unscaled graph $G$.}
\end{subfigure}
\begin{subfigure}[b]{.45\linewidth}
\centering
\begin{tikzpicture}[scale=0.5]
     \node[fill,circle, inner sep=0pt, minimum size=0.2cm] (a) at (-3,-3) {};
    \node[fill,circle, inner sep=0pt, minimum size=0.2cm] (c) at (-3,3) {};
    \node[fill,circle, inner sep=0pt, minimum size=0.2cm] (b) at (4,-3) {};
    \node[fill,circle, inner sep=0pt, minimum size=0.2cm] (d) at (4,3) {};
    \node[fill,circle, inner sep=0pt, minimum size=0.2cm] (r) at (0.5,6) {};
    \node at (-4,-3.7) {$a$};
    \node at (-4,3.7) {$c$};
    \node at (5,-3.7) {$b$};
    \node at (5,3.7) {$d$};
    \node at (0.5,6.7) {$r$};

    \path[->] (a) edge node[below,align = center] {\footnotesize {\color{black} $2$} }  (b);
   \path[->] (b) edge node[above,align = center, rotate=300] {\footnotesize {\color{black} $4$} }  (r);
   \path[->] (c) edge node[left,align = center] {\footnotesize {\color{black} $2$} }  (a);
   \path[->] (c) edge node[below,align = center] {\footnotesize {\color{black} $2$} }  (d);
   \path[->] (d) edge node[right,align = center] {\footnotesize {\color{black} $2$} }  (b);
   \path[->] (d) edge node[above,align = center, rotate=315] {\footnotesize {\color{black} $2$} }  (r);
   \path[->] (r) edge node[above,align = center, rotate= 60] {\footnotesize {\color{black} $4$} }  (a);
   \path[->] (r) edge node[above,align = center, rotate=45] {\footnotesize {\color{black} $-2$} }  (c);
\end{tikzpicture}
\subcaption{The scaled graph $\bar{G}$}
\end{subfigure}

\caption{An example of the scaling process when we have only one resource. The numbers on top of the edge denotes the resource consumption of that edge. In this example $\Delta = 2$. All edge weights are inherited from the original instance and therefore omitted in this graph to improve clarity.}
\label{fig:original_scaled_graph}
\end{figure}

Let $\sres(P) = \sum_{e \in P} \sres_{e}$ (recall that $\res(P) = \sum_{e \in P} \boldsymbol{r_e}$) for any walk $P$. For any $(s,t) \in D$, let $P_{(s,t)}$ be a feasible $s \leadsto t$ walk with the fewest number of edges. From our assumptions, $P_{(s,t)}$ has fewer than \hop edges. We have, 

\begin{align} 
     \sres(P_{(s,t)})[0] = \sum_{e \in P_{(s,t)}} \sres_{e}[0] \leq \sum_{e \in P_{(s,t)}} (\boldsymbol{r_e}[0] +\Delta).
     \notag
\end{align}

Since $P_{(s,t)}$ has $\leq \hop$ edges,

\begin{align} 
     \sres(P_{(s,t)})[0] \leq \left(\sum_{e \in P_{(s,t)}} \boldsymbol{r_e}[0] \right) + \hop \cdot (\Delta).
     \notag
\end{align}

Thus, we have,
\begin{align} \label{eq:horvathjunc_inequality}
     \sres(P_{(s,t)})[0] \leq \res(P_{(s,t)})[0] + \theta_1 \cdot \bdgt_{min} 
     \notag
     \\
     \leq \res(P_{(s,t)})[0] + \theta_1 \cdot \bdgt[0] \cdot \siign(\bdgt[0]).
\end{align}

Furthermore, since $\boldsymbol{r_e} \preceq \sres_e$, we have,
\begin{align} \label{eq:horvathjunc_inequality2}
     \res(P_{(s,t)}) \preceq \sres(P_{(s,t)}).
\end{align}

Roughly speaking, \eqref{eq:horvathjunc_inequality} and \eqref{eq:horvathjunc_inequality2} tell us that the resource consumptions of the scaled walks are close enough to the resource consumptions of their respective original walks (this is obvious for the last $m$ resources which undergo no change).

More formally, we have the following claim:
\begin{claim} \label{cl:junctree_scalingfeasible}
     If $P_{(s,t)}$ is a feasible solution for $(s,t) \in D$ (i.e., $\res^1(P_{(s,t)}) \preceq \bdgt(s,t)$), then $P_{(s,t)}$ is a $\theta$-feasible solution when we scale the weights (recall Definition \ref{def:theta_feasible_path}).

     Furthermore, if ${P_{(s,t)}}$ is a $\theta$-feasible solution in the scaled graph, then ${P_{(s,t)}}$ is a $\theta$-feasible solution in the unscaled/original graph.
\end{claim}
\begin{proof}

      If $P_{(s,t)}$ is a feasible solution for $(s,t) \in D$, then $\res^1(P_{(s,t)}) \leq \bdgt(s,t)$. Thus, using \eqref{eq:horvathjunc_inequality}, we can see that $\sres(P_{(s,t)})[0] \leq (1+\theta \cdot \siign(\bdgt(s,t)[0])) \cdot \bdgt(s,t)[0]$. Thus, $P_{s,t}$ is a $\theta$-feasible solution for the scaled version. 

      Now, let us look in the opposite direction. If $P_{(s,t)}$ is a $\theta$-feasible solution in the scaled graph, then $\sres(P_{(s,t)})[0] \leq (1+\theta \cdot \siign(\bdgt(s,t)[0])) \cdot  \bdgt(s,t)[0]$. Then using \eqref{eq:horvathjunc_inequality2}, we can see that,

      \begin{align} 
        \res^1(P_{(s,t)})[0] \leq \sres(P_{(s,t)})[0] \leq (1+\theta \cdot \siign(\bdgt(s,t)[0])) \cdot  \bdgt(s,t)[0] .
        \notag
    \end{align}
     which proves our Claim (Note that we don't need to prove anything for the last $m$ resources since they undergo no change).
\end{proof}

Next, we prove another claim that compares the cost of a partial solution in $\bar{G}$ with the cost of a partial solution in $G$.

\begin{claim} \label{cl:junctree_scalingoptimal}
    For any $f > 0$, and set of demand pairs $D' \subseteq D$, if there exists a sub graph $H$ in $G$ of total cost $\leq f$ containing a walk of resource consumption at most $\bdgt(s,t)$ from $s$ to $t$ for every $(s,t) \in D'$ then there exists a sub graph $\bar{H}$ in $\bar{G}$ of total cost $\leq f$ containing a $\theta$-feasible walk from $s$ to $t$ for every $(s,t) \in D'$.

    In addition, for any $f > 0$, if there exists a subgraph $\bar{H}$ in $\bar{G}$ of total cost $\leq f$ containing a $\theta$-feasible walk from $s$ to $t$ for every $(s,t) \in D'$ then there exists a subgraph $H$ in $G$ of total cost $\leq f$ containing a $\theta$-feasible walk from $s$ to $t$ for every $(s,t) \in D'$.
\end{claim}
\begin{proof}
   The first part of the claim can be proved easily. We can just set $\bar{H} = H$. Now, for any $(s,t) \in D$, $H$ has some walk $p_1(s,t)$ that is feasible. Then using the first part of Claim \ref{cl:junctree_scalingfeasible} we can see that $p_1(s,t)$ is a $\theta$-feasible walk for the same demand pair $(s,t)$ in $\bar{H}$. Since we are using the same set of walks and costs are inherited, the cost will remain the same.

    For the second part, we can again set $H = \bar{H}$. Now, for any $(s,t) \in D$, $\bar{H}$ has some walk $p_1(s,t)$ that is $\theta$-feasible. Using the second part of Claim \ref{cl:junctree_scalingfeasible} we can see that $p_1(s,t)$ is a $\theta$-feasible walk for the same demand pair $(s,t)$ in $H$. Since we are using the same set of walks and costs are inherited, the cost will remain the same.
\end{proof}

\subsubsection{Turning distance constraints into connectivity constraints}

\paragraph{High level idea and potential challenges:}

For a specific root vertex $r$, we turn our distance constraints with edge costs problem into a connectivity problem with edge costs. The overall process for this is as follows. We first build a product graph. After building this product graph, a number of dummy terminals are added to various layers - now the problem of finding a $s \leadsto r \leadsto t$ walk within the resource budget changes into a problem of connecting the correct dummy terminals. 

\paragraph{Graph construction:} Let us now see our graph construction. Let $\Delta$ be the scaling parameter associated with the given instance of \spks. Let $\boldsymbol{\delta}$ be a vector assigned as follows: $\boldsymbol{\delta}[0] = \Delta$, $\boldsymbol{\delta}[i] = 1 \; \forall i \in [1,m]$. Recall that we use the vectors $\boldsymbol{t^{-}}$ and $\boldsymbol{t^{+}}$ represent the smallest and largest possible multiple of $\boldsymbol{\delta}$ the resource consumption of a feasible walk could take. 

Observe that $\boldsymbol{t}^{-}[0] = \lfloor \min \bigl( \left( min_{e \in E} \boldsymbol{r_e}[1] \right) \cdot \hop / \Delta , 0 \bigr)  \rfloor$ - this happens when we take $\hop$ consecutive edges of scaled weight at least $min_{e \in E} \boldsymbol{r_e}[1]/\Delta$. Plugging in the value of $\Delta$, we have $\boldsymbol{t}^{-}[0] \geq (\eta \cdot \hop^2) / \theta = O((\eta \cdot n^2) / \theta)$ (Recall that $\hop \leq cn$). Now note that, $\boldsymbol{t}^{+}[0] = \lceil \bdgt_{max} \cdot (1+\theta) / \Delta \rceil + | \boldsymbol{t}^{-}[0] | $. This is because any feasible scaled walk has a length at most $\bdgt_{max} \cdot (1+\theta)$  - but a subwalk could be longer because it could decrease its length by as much as $t^{-} \cdot \Delta$ when it takes edges of negative length. Plugging in the value of $\Delta$, we have $\boldsymbol{t}^+[0] \leq \xi \cdot \hop /\theta + | \boldsymbol{t}^{-}[0] |$

For $i \in [1,m]$, 

\begin{itemize}
    \item when $i$ is a packing resource, $\boldsymbol{t}^{-}[i] = 0$ and $\boldsymbol{t}^{+}[i] = \threshold$.
    \item when $i$ is a covering resource, $\boldsymbol{t}^{+}[i] = 0$ and $\boldsymbol{t}^{-}[i] = -\threshold$.
\end{itemize}. 

We construct a product graph $\bar{G}_r$ with the following vertices:

\begin{align*}
         \Bar{V}_r^L = \left((V\setminus r) \times \prod_{i=0}^{m} \{\boldsymbol{t^{-}}[i] \cdot \boldsymbol{\delta[i]},\ldots,-2 \cdot \boldsymbol{\delta[i]} ,-1 \cdot \boldsymbol{\delta[i]},0 ,1 \cdot \boldsymbol{\delta[i]} ,2 \cdot \boldsymbol{\delta[i]} ,\ldots,\boldsymbol{t^{+}}[i] \cdot \boldsymbol{\delta[i]}\} \times \{L\}\right) \notag \cup \{(r,\zvec,L)\}, \
\end{align*}

\begin{align*}
         \Bar{V}_r^R = \left((V\setminus r) \times \prod_{i=0}^{m} \{\boldsymbol{t^{-}}[i] \cdot \boldsymbol{\delta[i]},\ldots,-2 \cdot \boldsymbol{\delta[i]} ,-1 \cdot \boldsymbol{\delta[i]},0 ,1 \cdot \boldsymbol{\delta[i]} ,2 \cdot \boldsymbol{\delta[i]} ,\ldots,\boldsymbol{t^{+}}[i] \cdot \boldsymbol{\delta[i]}\} \times \{R\}\right)\cup \{(r,\zvec,R)\},
\end{align*}

\begin{equation}
    \begin{split}
         \bar{V}_r = \bar{V}_r^R \cup \bar{V}_r^L.
    \end{split}
\end{equation}

As an example, a vertex in the newly constructed graph looks as follows: $(u,\boldsymbol{I} \cdot \boldsymbol{\delta}, L)$ \footnote{Throughout this work, given two vectors $\boldsymbol{a}$ and $\boldsymbol{b}$, $\boldsymbol{a} \cdot \boldsymbol{b}$ denotes a new vector $\boldsymbol{c}$ with $\boldsymbol{c}[i] = \boldsymbol{a}[i] \cdot \boldsymbol{b}[i]$}. This denotes that the new vertex is a copy of the vertex $u$ from the scaled graph $\bar{G}$, and this vertex can reach the root with a resource consumption of $\boldsymbol{I} \cdot \boldsymbol{\delta}$. Let us call this $\boldsymbol{I}$ (note that $\boldsymbol{I} \in Z ^ {m+1}$) as the label of the layer. We will explain the relevance of $L$ and $R$ later on. For now, it suffices to think of them as labels to distinguish two separate copies of the same vertex set.

We connect these vertices with the following edges:

\begin{equation}
    \begin{split}
        \Bar{E}_r^R = \{((u,\boldsymbol{I} \cdot \boldsymbol{\delta},R)(v,\boldsymbol{J} \cdot \boldsymbol{\delta},R)) | (u,\boldsymbol{I}\cdot \boldsymbol{\delta},R),(v,\boldsymbol{J}\cdot \boldsymbol{\delta},R) \in \bar{V}_r^R,(u,v)\in E \\
        \text{ and } \boldsymbol{\Bar{r}}_{(u,v)}=(\boldsymbol{J}-\boldsymbol{I} )\cdot \boldsymbol{\delta} \text{ where } \boldsymbol{I},\boldsymbol{J} \in  Z^m \text{ and } \boldsymbol{\Bar{r}}_{(u,v)} \text{ is the scaled length of the edge} (u,v)\}.
    \end{split}
\end{equation}

\begin{equation}
    \begin{split}
        \Bar{E}_r^L = \{((u,\boldsymbol{I} \cdot \boldsymbol{\delta},L)(v,\boldsymbol{J} \cdot \boldsymbol{\delta},L)) | (u,\boldsymbol{I}\cdot \boldsymbol{\delta},L),(v,\boldsymbol{J}\cdot \boldsymbol{\delta},L) \in \bar{V}_r^L,(u,v)\in E \\
        \text{ and } \boldsymbol{\Bar{r}}_{(u,v)}=(\boldsymbol{I} - \boldsymbol{J} )\cdot \boldsymbol{\delta} \text{ where } \boldsymbol{I},\boldsymbol{J} \in  Z^m \text{ and } \boldsymbol{\Bar{r}}_{(u,v)} \text{ is the scaled length of the edge} (u,v)\}.
    \end{split}
\end{equation}

Intuitively, we add an edge between two vertices whenever it makes sense i.e. when the original copies of these two vertices are connected in the scaled graph $\bar{G}$ and the layer separation between these two vertices is equal to the resource consumption of the scaled edge in  $\bar{G}$. The edges in our product graph $\Bar{G}_r$ inherit the costs from the corresponding edges in the original graph. The edges and vertices are built in such a way that if there is a walk from $(u,\boldsymbol{I} \cdot \boldsymbol{\delta},L)$ to the root $r$, then in the scaled graph there is a $u \leadsto r$ walk $p(u,r)$ such that $\res(p(u,r)) = \boldsymbol{I} \cdot \boldsymbol{\delta}$. Similarly, if there is a walk from the root $r$ to $(u,\boldsymbol{I} \cdot\boldsymbol{\delta},R)$, then in the scaled graph there is a $r \leadsto u$ walk $p(r,u)$ such that $\res(p(r,u)) = \boldsymbol{I} \cdot \boldsymbol{\delta}$.

We add one final edge. This connects $(r,\zvec,L)$ to $(r,\zvec,R)$. It is a dummy edge and therefore it has zero cost.

\begin{equation}
    \Bar{E}_r = \Bar{E}_r \cup \{((r,\zvec,L),(r,\zvec,R))\}.
\end{equation}

Let,

\begin{equation}
    \Bar{E}_r = \Bar{E}_r^R \cup \Bar{E}_r^L.
\end{equation}

\begin{equation}
    \bar{G}_r = (\bar{V}_r,\bar{E}_r).
\end{equation}

\begin{definition} \label{de:valid_layer_complex}
    For a given $\theta_2 \in R^+$, we call a vector $\boldsymbol{I} \in Z^{m+1}$ {\em valid} if $\boldsymbol{t^{-}} \preceq \boldsymbol{I} \preceq \boldsymbol{t^{+}}$.
\end{definition} 

For every terminal pair $(s,t) \in D$, do the following, 

\begin{enumerate}
    \item Add new vertices $(s^t,\boldsymbol{I} \cdot \boldsymbol{\delta})$ and $(t^s,\boldsymbol{J} \cdot \boldsymbol{\delta})$ for all {\em valid} vectors $\boldsymbol{I},\boldsymbol{J}$  to $\bar{V}_r$.
    \item For all such $\boldsymbol{I}$ and $\boldsymbol{J}$ add edges $((s^t,\boldsymbol{I} \cdot \boldsymbol{\delta})(s,(\boldsymbol{I} \cdot \boldsymbol{\delta},L)))$ and $((t,(\boldsymbol{J} \cdot \boldsymbol{\delta},R))(t^s,\boldsymbol{J} \cdot \boldsymbol{\delta}))$ with zero cost to $E_r$.
    \item Now for every terminal pair $(s,t) \in P$ define:
    \begin{enumerate}
        \item terminal sets $S_{s,t} = \{(s^t,\boldsymbol{I} \cdot \boldsymbol{\delta}) \:\forall \text{ {\em valid} vectors } I\}$,
        \item $T_{s,t} = \{(t^s,\boldsymbol{J} \cdot \boldsymbol{\delta}) | \:\forall \text{ {\em valid} vectors }J\}$ and 
        \item relation $R_{s,t} = \{(s^t,\boldsymbol{I}\cdot \boldsymbol{\delta}),(t^s,\boldsymbol{J}\cdot \boldsymbol{\delta}) \in S_{s,t}\times T_{s,t} | (\boldsymbol{I}+\boldsymbol{J}) \cdot \boldsymbol{\delta} \preceq \bdgt(s,t) $.
    \end{enumerate}
\end{enumerate}

As in Section \ref{sec:jtree_vanilla_version}, the above construction does not handle one specific case where a feasible walk consumes fewer resources than $t^-[i]$. 

Observe that, since a covering resource always has $\boldsymbol{r_e}[i] \leq 0$, once a walk $p$ has $\res(p)[i] \leq \bdgt[i]$, adding more edges to the walk will never make $\res(p)[i] > \bdgt[i]$. So, if a walk has $\res(p)[i] \leq \bdgt[i]$, we don't need to track $\res(p)[i]$ any more. We just need to ensure that the walk leads to the root i.e., at this stage, we only need to preserve connectivity not the resource constraint. Thus, we don't need to build any extra layers beyond $\bdgt[i] >= t^-[i]$ for the $i^{th}$ resource. To remedy the issue mentioned in the previous paragraph, we add another set of edges. For every covering resource $i$, every valid vector $\hat{I}$ of the form $(\ldots,t^-[i],\ldots)$ and for every edge $(u,v) \in E$, we add an edge from $(u,\hat{I})$ to $(v,\hat{I})$ to $\bar{E_r}$. These edges will allow us to track walks with $\res(p)[i] \leq -\tau \leq \bdgt[i]$.

\paragraph{Relating the product graph with the scaled graph:}

Let us now create a simpler graph: Let $\hat{G}$ be a graph comprised of two copies of $\bar{G}$ - namely $G_{-}$ and $G_{+}$ intersecting in the node $r$. For every vertex $u \in V$, let $u_{+}$ and $u_{-}$ denote the copies of $u$ in $G_{+}$ and $G_{-}$ respectively. We call this graph as the {\em intersection graph}. Let $|\validlayer| = \prod_{i=1}^{m} \left(O(\boldsymbol{t^{-}}[i]+ \boldsymbol{t^{+}}[i] \right)$.

The following two claims will relate the product graph $\bar{G_r}$ with the intersection graph $\hat{G}$ (and thus indirectly the scaled graph $\bar{G}$).

\begin{claim} 
    For any $f > 0$, and set of terminal pairs $D' \subseteq D$, assume there exists a subgraph $\hat{H}$ in $\hat{G}$ of total cost $\leq f$ containing a walk of resource consumption at most $\bdgt(s,t)$ from $s_+$ to $t_-$ for every $(s,t) \in D'$. Then there exists a subgraph $\bar{H_r}$ in $\bar{G_r}$ of total cost $\leq f \cdot |\validlayer|$ containing a walk from $(s^t, \boldsymbol{I} \cdot \boldsymbol{\delta})$ to $(t^s, \boldsymbol{J} \cdot \boldsymbol{\delta})$ such that $((s^t,\boldsymbol{I} \cdot \boldsymbol{\delta}),(t^s,\boldsymbol{J}  \cdot \boldsymbol{\delta})) \in R_{s,t}$ for every $(s,t) \in D'$. 

    \end{claim}
    
\begin{proof}
     Let $p(s_i,t_i)$ be a $s_i \leadsto t_i$ walk in $\hat{H}$. Due to the construction of $\hat{G}$, this walk is forced to go through the root $r$ and $\res(p(s_i,t_i) \leq \bdgt(s,t)$. Let $p(s_i,r)$ be a subwalk of $p(s_i,t_i)$ that ends at $r$ and let $\res(p(s_i,r)) = \boldsymbol{I} \cdot \boldsymbol{\delta}$. By construction $\boldsymbol{t^-} \preceq \boldsymbol{I} \preceq \boldsymbol{t^+}$. Similarly, let $p(r,t_i)$ be a $r \leadsto t_i$ subwalk of $p(s_i,t_i)$ and let $\res(p(s_i,r)) = \boldsymbol{J} \cdot \boldsymbol{\delta}$. By construction $\boldsymbol{t^-} \preceq \boldsymbol{J} \preceq \boldsymbol{t^+}$. 
        
    Observe that $\boldsymbol{J} + \boldsymbol{I} \preceq \bdgt(s,t)$. Additionally,  by construction $(s^t,\boldsymbol{I} \cdot \boldsymbol{\delta}),(t^s,\boldsymbol{J} \cdot \boldsymbol{\delta}) \in \bar{V}_r$ and thus $((s^t,\boldsymbol{I}),(t^s,\boldsymbol{J})) \in R_{s,t}$. Note that $(s^t,\boldsymbol{I} \cdot \boldsymbol{\delta})$ and $(t^s,\boldsymbol{J} \cdot \boldsymbol{\delta})$ can be connected by a walk analogous to $p(s_i,t_i)$. Compile all such walks and we will have $\bar{H_r}$.

    For the cost note that we effectively use the same set of edges but since $\bar{G_r}$ has multiple copies of the edges from $\bar{G}$, it could over count some edges. This issue (which does not arise in \cite{chlamtavc2020approximating}) is why we need to be careful about how many layers we build. In \cite{chlamtavc2020approximating}, we have $m=0$ and therefore the layers we include have a dominating set of size $1$ (i.e., when we are trying to reach the root, there is no need to consider an edge at a distance of $i$ from the root when we include a different copy of the same edge at a distance of $j$ from the root with $j < i$). That is not the case when $m \geq 1$. 
    
    For any layer $\boldsymbol{I}$, we call the assignment of the last $m$ elements of $\boldsymbol{I}$ as its configuration. We now note that for a specific configuration, we need only one copy of an edge. This is because when we fix a configuration, a specific assignment of the zeroth resource will dominate all other assignments (and thus we have a dominating set of size $1$ after we fix a configuration). In total, the number of configurations we have is given by $|\validlayer|$ and that proves the rest of the claim.
\end{proof}

\begin{claim}

    For any $f > 0$, assume there exists a subgraph $\bar{H_r}$ in $\bar{G_r}$ of total cost $\leq f$ containing a walk from from $(s^t, \boldsymbol{I} \cdot \cdot \boldsymbol{\delta})$ to $(t^s, \boldsymbol{J} \cdot \cdot \boldsymbol{\delta})$ such that $((s^t,\boldsymbol{I} \cdot \boldsymbol{\delta}),(t^s,\boldsymbol{J} \cdot \boldsymbol{\delta})) \in R_{s,t}$ for every $(s,t) \in D'$. Then there exists a subgraph $\hat{H}$ in $\hat{G}$ of total cost $\leq f$ containing a walk of resource consumption at most $\bdgt(s,t) $ from $s_+$ to $t_-$ for every $(s,t) \in D'$.
\end{claim}

\begin{proof}
    This claim is much simpler. Let $p(s_i,t_i)$ be a $s_i \leadsto t_i$ walk in $\bar{H_r}$. Due to the construction of $\bar{G_r}$, this walk is forced to go through the root $r$ and $\res(p(s_i,t_i) \leq \bdgt(s,t)$. Due to the construction of $\bar{G_r}$, we can see that there is a $s_i \leadsto r \leadsto t_i$ walk $p'$ in $\hat{G}$ with the same cost and resource consumption $\preceq \res(p(s_i,t_i)) \preceq \bdgt(s,t)$. 

    Compile all such walks and we have $\hat{H}$. The overall cost of this set $\leq f$ since we can reuse the same set of edges (in practice the cost can be lower since $\hat{G}$ has only two duplicate copies of any edge unlike $\bar{G_r}$ which has several duplicate copies).
\end{proof}

\paragraph{Runtime:} 
Our runtime so far is polynomial in $|\bar{G}_r|$ (for any graph $G$, we use $|G|$ to denote the number of vertices in $G$). Note that $|\bar{G}_r|$ is a polynomial in $O(n \cdot \Pi_{i=0}^{i=m}(|\boldsymbol{t^{-}}[i]| + |\boldsymbol{t^{+}}[i]| ))$.

\paragraph{To summarize,} we have turned all budget constraints into connectivity constraints so far. We keep track of the budget constraints using some relations. 

The rest of this proof is identical to the proof of Theorem \ref{thm:mdrcjt_vanilla} and is therefore omitted.

\section{An approximation algorithm for \pks (Rational version)}
\label{sec:khalf_complex}
In this section we see a proof of Theorem \ref{thm:khalf_pks}.

\begin{restatable}{theorem}{thmkhalfpks} \label{thm:khalf_pks}

Let $k = |D|$ where $D$ is the set of demand pairs in the given instance of \pks.

Then, for any $\theta > 0$ and any two constants $\eps,m > 0$, there is a $\poly(n,\xi,\eta,1/\theta,\tau^m)$-time algorithm for \pks with an approximation ratio $\tilde{O}(k^{1/2 + \ep} \threshold^m \cdot \polylog(\tau^m \cdot ((\eta+\xi)/\theta)))$  returning a subgraph $H$ that contains a $\theta$-feasible walk $p(s,t)$ for each $(s,t) \in D$. Here, $\eta$ and $\xi$ are the condition numbers defined in \eqref{def:eta} and \eqref{def:xi}, respectively. When $1/\theta, \eta, \xi, \tau \in \poly(n)$, the algorithm runs in polynomial time.
\end{restatable}

\begin{proof}
    The proof closely follows that for Theorem \ref{thm:vpks_khalf}. The main difference is that we use $\theta$-relaxed resource-constrained junction tree solutions.

    \begin{definition}
A \emph{$\theta$-relaxed resource-constrained junction tree solution} is a collection of $\theta$-relaxed resource-constrained junction trees rooted at different vertices, such that there exists a $\theta$-feasible walk (recall Definition \ref{def:theta_feasible_path}) for all $(s,t) \in D$.
\end{definition}

Our proof has two main ingredients. First, 
we construct a $\theta$-feasible resource-constrained junction tree solution and compare its objective with the optimal $\theta$-feasible resource-constrained junction tree solution with objective value $\opt_{junc}$. Theorem \ref{thm:mdrcjt_complex} implies a $\poly(n,\xi,\eta,1/\theta,\tau^m,1/\ep)$-time algorithm that finds a $\theta$-feasible resource-constrained junction tree solution of cost at most $\tilde{O}(k^{\ep} \threshold^m \cdot \polylog(\tau^m \cdot ((\eta+\xi)/\theta))) \opt_{junc}$. Second, we show the existence of an $O(\sqrt{n})$-approximate solution consisting of resource-constrained junction trees, i.e., $\opt_{junc} \le O(\sqrt{n}) \opt$ where $\opt$ is the optimal cost when the resource constraints are strict. Combining these two ingredients implies Theorem \ref{thm:khalf_pks}.

We use a density argument as in the proof for Theorem \ref{thm:vpks_khalf}. The main difference is that here we consider $\theta$-relaxed resource-constrained junction trees. Let $\opt_{\theta}$ denote the cost of the optimal solution where the resource constraints are $\theta$-relaxed. Clearly, $\opt_\theta \le \opt$ because the resource constraints for $\opt$ are stricter. The following Lemma follows the same way as Lemma \ref{lem:sqrt-k-den}.

\begin{restatable}{lemma}{lemsqrtkdentheta}
There exists a $\theta$-relaxed resource-constrained junction tree $\cJ$ with density at most $\opt_\theta / \sqrt{k}$.
\end{restatable}

Using this lemma and the iterative procedure as in the proof of Theorem \ref{thm:vpks_khalf} implies $\opt_{junc} \le O(\sqrt{k}) \opt_\theta \le O(\sqrt{k}) \opt$, which implies Theorem \ref{thm:khalf_pks}.
\end{proof}

\newpage
\section{Acknowledgments}

We thank the anonymous reviewers for their valuable suggestions and comments, which helped us improve the quality of this writeup. We are also grateful to Kent Quanrud for pointing out that the constraints in \pks can be interpreted as packing and covering constraints, which inspired us to re-frame the problem in its current form.

\bibliographystyle{acm}
\bibliography{reference}

\section{List of Symbols and Notation}\label{sec:appendix}

We present here the main notation used throughout our paper. This list is cumulative, so notation introduced in earlier sections (e.g., Section~\ref{sec:intro}) may also appear in later ones.

\subsection{Notation from Section \ref{sec:intro}}

\begin{itemize}
    \item $\tau$ — magnitude threshold.  
    \item $m$ — number of resources in an instance of \pks or number of sets in an instance of \rcs.  
    \item $k$ — number of demand pairs.  
    \item $n$ — number of nodes in the graph.  
    \item $\mathbf{r}_e$ — resource consumption vector of edge $e$.  
    \item $\bdgt(s,t)$ — resource budget for demand pair $(s,t)$.  
    \item $D$ — set of demand pairs.  
    \item $\sigma(e)$ — cost of edge $e$.  
    \item $\ctrl(s,t)$ — control vector in \rcs, encoding distance, must-visit, and must-avoid requirements.  
    \item $\ell_e$ — length of edge $e$.  
\end{itemize}

\subsection{Notation from Section \ref{sec:jtree_vanilla_version}}

\begin{itemize}
    \item $\mathbf{t}^{-}, \mathbf{t}^{+}$ — componentwise lower and upper bounds on the resource consumption of any feasible subpath.  
    \item $\bar{G}_r$ — product graph rooted at $r$, constructed from $G$.  
    \item $\bar{G}_r^R, \bar{G}_r^L$ — right (edges directed away from root) and left (edges directed towards root) components of $\bar{G}_r$.  
    \item $\hat{G}$ — intersection graph, consisting of two copies of $G$, namely $G_{-}$ and $G_{+}$, intersecting at $r$.  
    \item $|\validlayer| = \prod_{i=1}^m O(\mathbf{t}^-[i] + \mathbf{t}^+[i])$ — number of valid layer configurations.  
    \item $\bar{G}_r^{\mathrm{up}}$ — layered graph from applying height reduction to $\bar{G}_r^R$.  
    \item $\bar{G}_r^{\mathrm{down}}$ — layered graph from applying height reduction to $\bar{G}_r^L$.  
    \item $\bar{T}_r$ — graph obtained by connecting the root $r^{\mathrm{up}}$ in $\bar{G}_r^{\mathrm{up}}$ to the root $r^{\mathrm{down}}$ in $\bar{G}_r^{\mathrm{down}}$.  
    \item $\Psi : V(\bar{T}_r) \to V(\bar{G}_r)$ — mapping between vertices of the height-reduced graph and the product graph.  
    \item $S_{s,t}^r = \{\Psi^{-1}((s^t,\mathbf{I})) \mid (s^t,\mathbf{I}) \in S_{s,t}\}$ — set of terminals in $\bar{T}_r$ corresponding to $S_{s,t}$.  
    \item $T_{s,t}^r = \{\Psi^{-1}((t^s,\mathbf{I})) \mid (t^s,\mathbf{I}) \in T_{s,t}\}$ — set of terminals in $\bar{T}_r$ corresponding to $T_{s,t}$.  
    \item $R_{s,t}^r = \{(s^t_r,\mathbf{I}),(t^s_r,\mathbf{J}) \in S_{s,t}^r \times T_{s,t}^r \mid \mathbf{I} + \mathbf{J} \preceq \bdgt(s,t)\}$.  
\end{itemize}

\subsection{Notation from Section \ref{sec:khalf_vanilla}}

\begin{itemize}
    \item $\opt_{\mathrm{junc}}$ — cost of an optimal resource-constrained junction tree.  
\end{itemize}

\subsection{Notation from Section \ref{sec:jtree_complex_version}}

\begin{itemize}
    \item Condition numbers:  
    \[
        \eta := \frac{\left|\min\left\{\min_{e \in E}\{\mathbf{r}_e[1]\},0\right\}\right|}{\min_{(s,t) \in D} |\bdgt(s,t)[0]|}, \quad
        \xi := \frac{\max_{(s,t) \in D} |\bdgt(s,t)[0]|}{\min_{(s,t) \in D} |\bdgt(s,t)[0]|}.
    \]
    \item $\theta$ — tolerance parameter (see Definition~\ref{def:theta_feasible_path}).  
    \item $\bdgt_{\max} = \max_{(s,t) \in D} |\bdgt(s,t)[0]|$.  
    \item $\bdgt_{\min} = \min_{(s,t) \in D} |\bdgt(s,t)[0]|$.  
    \item $\lenlow = \min_{e \in E} \mathbf{r}_e[0]$.  
    \item $\lenup = \max_{e \in E} \mathbf{r}_e[0]$.  
    \item $\hop$ — smallest positive integer such that for every $(s,t) \in D$, there exists a feasible $s \leadsto t$ walk with fewer than $\hop$ edges.  
    \item $\Delta = \theta \cdot \bdgt_{\min}/\hop$.  
    \item $\bar{G}$ — scaled version of $G$.  
\end{itemize}

\end{document}